\documentclass[11pt]{amsart}
\usepackage[margin=1in]{geometry}
\usepackage{amsfonts}
\usepackage{amssymb}
\usepackage[dvips]{graphics}
\usepackage{epsfig}
\usepackage{euscript}
\usepackage{color}

 \newtheorem{thm}{Theorem}[section]
 \newtheorem{cor}[thm]{Corollary}
 \newtheorem{lemma}[thm]{Lemma}
 \newtheorem{prop}[thm]{Proposition}
 \theoremstyle{definition}
 
 \theoremstyle{remark}
 \newtheorem{rem}[thm]{Remark}
 
 \newtheorem{assumption}[thm]{Assumption}
 \numberwithin{equation}{section}
 
 \def\idty{{\mathchoice {\mathrm{1\mskip-4mu l}} {\mathrm{1\mskip-4mu l}} %
{\mathrm{1\mskip-4.5mu l}} {\mathrm{1\mskip-5mu l}}}}

\newcommand{\bR}{{\mathbb R}}

\newcommand{\cA}{{\mathcal A}}
\newcommand{\cB}{{\mathcal B}}
\newcommand{\cH}{{\mathcal H}}

\newcommand{\supp}{\operatorname{supp}}
\newcommand{\cU}{{\mathcal U}}

\newcommand{\caA}{{\mathcal A}}
\newcommand{\caB}{{\mathcal B}}

\newcommand{\caH}{{\mathcal H}}

\newcommand{\caP}{{\mathcal P}}

\newcommand{\bbC}{{\mathbb C}}

\newcommand{\bbR}{{\mathbb R}}

\newcommand{\ie}{{\it i.e.\/} }

\newcommand{\iu}{\mathrm{i}}

\newcommand{\str}{^{*}}
\newcommand{\ep}[1]{\mathrm{e}^{#1}}
\newcommand{\hilb}{\mathcal{H}}

\newcommand{\Tr}{\mathrm{Tr}}

\newcommand{\teop}{\hfill$\square$}
\newcommand{\ran}{{\rm ran}}

\newcommand{\ketbra}[1]{\left\vert #1\right\rangle\!\!\left\langle #1\right\vert}
\newcommand{\kettbra}[2]{\left\vert #1\right\rangle\!\!\left\langle #2\right\vert}

\newcommand{\be}{\begin{equation}}
\newcommand{\ee}{\end{equation}}
\newcommand{\bea}{\begin{eqnarray}}
\newcommand{\eea}{\end{eqnarray}}
\newcommand{\beann}{\begin{eqnarray*}}
\newcommand{\eeann}{\end{eqnarray*}}

\newcommand{\Rl}{\bR}

\newcommand{\eq}[1]{(\ref{#1})}

\newcommand{\id}{{\rm id}}

\begin{document}

\renewcommand{\thefootnote}{\fnsymbol{footnote}}
\title[Automorphic Equivalence]{Automorphic Equivalence within Gapped Phases\\
of Quantum Lattice Systems}

\author[S. Bachmann]{Sven Bachmann}
\address{Department of Mathematics\\
University of California, Davis\\
Davis, CA 95616, USA}
\email{svenbac@math.ucdavis.edu}

\author[S. Michalakis]{Spyridon Michalakis}
\address{Institute for Quantum Information\\
California Institute of Technology\\
Pasadena, CA 91125, USA}
\email{spiros@caltech.edu}

\author[B. Nachtergaele]{Bruno Nachtergaele}
\address{Department of Mathematics\\
University of California, Davis\\
Davis, CA 95616, USA}
\email{bxn@math.ucdavis.edu}

\author[R. Sims]{Robert Sims}
\address{Department of Mathematics\\
University of Arizona\\
Tucson, AZ 85721, USA}
\email{rsims@math.arizona.edu}

\date{\today }

\begin{abstract}

%Recently, there has been renewed interest in gapped quantum phases
%and quantum models which exhibit phase transitions. 
Gapped ground
states of quantum spin systems have been referred to in the
physics literature as being `in the same phase' if there exists a
family of Hamiltonians $H(s)$, with finite range interactions
depending continuously on $s\in [0,1]$, such that for each $s$,
$H(s)$ has a non-vanishing gap above its ground
state and with the two initial states being the ground states of $H(0)$ and $H(1)$,
respectively. In this work, we give precise conditions
under which any two gapped ground states of a given quantum spin
system that 'belong to the same phase' are automorphically equivalent and
show that this equivalence can be implemented
as a flow generated by an $s$-dependent interaction
which decays faster than any power
law (in fact, almost exponentially). 
The flow is constructed using Hastings' `quasi-adiabatic evolution'  technique,
of which we give a proof extended to infinite-dimensional Hilbert spaces.
In addition, we derive a general result about the locality
properties of the effect of perturbations of the dynamics for quantum
systems with a quasi-local structure and prove that the flow, which we
call the {\em spectral flow}, connecting the gapped ground states in the same phase,
satisfies a Lieb-Robinson bound. As a result, we obtain that, in the thermodynamic
limit, the spectral flow converges to a co-cycle of automorphisms of the algebra of
quasi-local observables of the infinite spin system. This proves that the ground
state phase structure is preserved along the curve of models $H(s)$, $0\leq s\leq 1$.
\end{abstract}

\maketitle

\footnotetext[1]{Copyright \copyright\ 2011 by the authors. This
paper may be reproduced, in its entirety, for non-commercial
purposes.}

\section{Introduction}\label{sec:intro}

Since the discovery of the fractional quantum Hall effect \cite{tsui:1982} and its description
in terms of model wave functions with special `topological' properties \cite{laughlin:1983},
there has been great interest in quantum phase transition \cite{sachdev:2000}. Experimental 
and theoretical discoveries of exotic states in strongly correlated systems \cite{dagotto:1996} 
and, more recently, the possibility of
using topologically ordered quantum phases for quantum information computation 
\cite{kitaev:2003}, have further increased our need to understand the nature of quantum phase 
transitions, and especially of gapped ground states. It is natural to ask whether
gapped quantum phases and the transitions between them can be classified. The first
and simplest question is to define precisely what it means for two gapped ground states
to belong to the same phase. A pragmatic definition that has recently been considered
in the literature declares two gapped ground states of a quantum spin system to belong
to the same phase if there exists a family of Hamiltonians $H(s)$, with finite range
interactions depending continuously on $s\in [0,1]$, such that for each $s$, 
$H(s)$ has a non-vanishing gap above its ground state, and the two given 
states are the ground states of $H(0)$ and $H(1)$. In other words there is a family of
Hamiltonians with gapped ground states that interpolate between the given two
\cite{chen:2010a,chen:2010b}. In this paper we prove a result that supports this definition.
We show that any two gapped ground states in the same phase according to this definition
are unitarily equivalent, with a unitary that can be obtained as the flow of an 
$s$-dependent quasi-local interaction which decays almost exponentially fast.
When applied to models on a finite-dimentsional lattice, this quasi-local structure is 
sufficient to prove that the unitary equivalence of finite volume leads to automorphic 
equivalence at the level of the $C^*$-algebra of quasi-local observables in the 
thermodynamic limit.

In statistical mechanics, lattice models with short-range interactions
play a central role. Many examples of Hamiltonians that can be considered 
as a perturbation of a model with a known ground state that is sufficiently 
simple (typically given by finite number of classical spin configurations), 
have been studied by series expansion methods \cite{kennedy:1985, albanese:1990a,
albanese:1990b, matsui:1990, kennedy:1992a, kennedy:1992b, borgs:1996, datta:1996a, 
datta:1996b,yarotsky:2006}. Perturbation expansions
provide a detailed understanding of the ground state under quite 
general conditions. Often, one can prove the existence of a finite correlation
length and  a non-vanishing spectral gap above the ground state, and short-range
correlations can in principle be calculated to arbitrary precision.
The perturbation series one employs in such situations has the structure of a 
cluster expansion in which the cluster geometry is  based
on the underlying lattice structure and the short-range nature of the interactions.
The effect of the perturbations can then be understood as approximately local
modifications of the ground state of the unperturbed model.

In this work, we start from a different perspective. Suppose we have a family of models
defined in terms of an interaction $\Phi(s)$ which depends on a parameter $s\in [0,1]$:
\begin{equation}\label{Hs}
H_\Lambda(s)=\sum_{X\subset\Lambda} \Phi(X,s)
\end{equation}
Here, $\Lambda$ is a finite subset of the lattice $\Gamma$ (e.g., $\Gamma=\mathbb{Z}^\nu$),
and $\Phi(s)$ is a short-range interaction depending smoothly on a parameter $s$ (see
Section~\ref{sec:quantum_phases} for the precise conditions on the decay of the interactions
that we assume).

Suppose that for all $s$ the ground state(s) of this family of models are isolated
from the rest of the spectrum by a gap. 
We prove that the ground state(s) of $H_\Lambda(s)$ can be obtained from the
ground state(s) of $H_\Lambda(0)$ by a unitary transformation $U_\Lambda(s)$ which
has a quasi-local structure in the sense that $U_\Lambda(s)$ can be regarded
as the flow generated by a quasi-local parameter-dependent interaction
$\Psi(s)$ which we construct. In the works cited above the goal was to develop 
a suitable perturbation theory 
which would allow one to prove the existence of a non-vanishing spectral gap, among 
other things. To do that one has to start from a sufficiently simple model at $s=0$ 
and also assume that the perturbation potential is sufficiently small.
Note that no such smallness condition is required on $\Phi(s)$ here.
We now make some comments on the methods used in this paper.

In his 2004 paper  \cite{hastings2004a}
Hastings introduced a new technique, which he called `quasi-adiabatic continuation'
(see also \cite{hastings05}).
He used it in combination with the propagation bounds for quantum lattice 
dynamics originally due to Lieb and Robinson \cite{lieb1972} to construct and 
analyze the variational states needed for the proof of a multi-dimensional version
of the Lieb-Schultz-Mattis theorem \cite{lieb1961}.
The quasi-adiabatic continuation technique was subsequently elaborated upon and 
used in new applications by Hastings and collaborators 
\cite{bravyi:2006,hastings2007,hastings09,bravyi:2010b,bravyi:2010a} as well as other
authors \cite{Nach07,Os2007}. In this paper we give a general account of this technique and 
show how it allows one to exploit locality properties of the dynamics of extended 
quantum systems with short-range interactions without resorting to cluster expansions.

The starting point of the analysis in all the works mentioned above is a version
of our Proposition \ref{prop:main}.
This result shows that the spectral projection associated with an isolated part of 
the spectrum of a family of self-adjoint operators $H(s)$ depending smoothly on 
a parameter $s$, can be obtained as a unitary evolution.
Concretely, let $I(s)\subset \Rl$ denote an interval such that for all $s$ the spectrum
of $H(s)$ contained in $I(s)$ is separated by a uniform gap $\gamma>0$ from the
rest of the spectrum of $H(s)$, then there exists a curve of unitary operators
$U(s)$ such that 
$$
P(s) = U(s) P(0) U(s)^*
$$
When we apply this result to families of Hamiltonians
$H_\Lambda(s)$ of the form (\ref{Hs}), i.e., with a quasi-local structure, we find
that the unitaries $U_\Lambda(s)$ then have the structure of a quasi-local dynamics itself.
Explicitly,
$$
\frac{d}{ds} U_\Lambda(s) = iD_\Lambda(s) U_\Lambda(s)
$$
where $D_\Lambda(s)$ is a self-adjoint operator with the structure of a time-dependent
Hamiltonian,\ie, there is an interaction $\Psi(s)$ such that
$$
D_{\Lambda}(s)=\sum_{X\subset \Lambda}\Psi(X,s)\, .
$$
Because of this quasi-local structure, the flow on the algebra of observables
defined by conjugation with the unitaries $U_\Lambda(s)$, \ie:
\begin{equation}\label{spectralflow}
\alpha^\Lambda_s(A) =  U_\Lambda^*(s) AU_\Lambda(s)
\end{equation}
satisfies a propagation bound of Lieb-Robinson type (see Section \ref{Sec:LRforD}). 
These propagation bounds---as a second application of Lieb-Robinson bounds in
this paper--- can be used to prove the existence of the thermodynamic limit
(Section \ref{sec:quantum_phases}). 
The main result of this paper is Theorem \ref{thm:main}. Stated in words, it says that
if for a differentiable curve of Hamiltonians of the form (\ref{Hs}) the gap above the
ground states does not close along the curve, then, for each $s$ there is an automorphism
$\alpha_s$ of the algebra of quasi-local observables which maps the ground states
at $s=0$ to the ground states at $s$. In particular the simplex
of infinite-volume ground states for all values of $s$ is isomorphic to the one for $s=0$.

We find the designation `quasi-adiabatic' of the flow $\alpha^\Lambda_s$
somewhat misleading since there is nothing adiabatic about it. The flow does,
however, follow the spectral subspace belonging to the isolated interval
$I(s)$. We will therefore call it the {\em spectral flow}.

This paper is organized as follows. In Section \ref{sec:Evolution}, we give a
rigorous and self-contained presentation of the construction of the spectral
flow in a form that allows for applications with an infinite-dimensional Hilbert space. 
A number of applications where the infinite-dimensional context has proven useful have 
already been considered in the literature, see e.g. 
\cite{cramer:2008, nachtergaele2009, Nach10, Amour09,PS09,PS10}. We expect that more applications will be found. 

In Section \ref{Sec:LocalPert} we use Lieb-Robinson bounds to obtain
a locality property of the spectral flow and prove that
{\em local perturbations perturb locally}
in the sense that the dependence of gapped ground states (or any other isolated 
eigenstates) on any given local term in the Hamiltonian is significant only in a 
neighborhood of the support of that term. In Lemma \ref{lma:Approx}
we generalize the notion of normalized partial trace to infinite-dimensional Hilbert 
spaces.

In the final two sections we consider quantum lattice models, or more generally,
models defined on a metric graph (satisfying suitable conditions) with sufficiently 
fast decaying interactions. Section \ref{Sec:LRforD}
is devoted to showing that the spectral flow
can be generated by time-dependent Hamiltonians defined in terms of 
local interactions. As a consequence, this flow then also satisfies a Lieb-Robinson
bound. In Section \ref{sec:quantum_phases}, we restrict our attention to 
quantum spin systems, and use the results of Section \ref{Sec:LRforD}
to obtain the existence of the thermodynamic limit of the spectral flow
as automorphisms on the algebra of quasi-local observables. We conclude
the paper with a brief discussion of the notion of `gapped ground state phase',
which has been a topic of particular interest in the recent literature.

%%%%%%%%%%%%%%%%%%%%%%%%%%%%%%%%%%%%%%%%%%%%%%%%%%%%%%%%%%%%%%%%%%%%%%%%%%%%%%%%%%%%%%%%%%%%%%%%%%%%%%%%

\section{The curve of spectral projections for an isolated part of the spectrum of a Hamiltonian
with a parameter}
\label{sec:Evolution}

We consider a smooth family of self-adjoint Hamiltonians $H(s)=H(s)\str$ parametrized by $s\in[0,1]$, acting on a Hilbert space $\hilb$. We do not assume that $H(s)$ itself is bounded but
the $s-$dependent portion should be. We are interested in the spectral projection
$P(s)$ associated with an isolated part of the spectrum of $H(s)$. Explicitly, our main
assumption on $H(s)$ is the following.

\begin{assumption}\label{Assump:Hbd}
$H(s)$ is a densely defined self-adjoint operator with bounded derivative $H^\prime(s)$,
such that $\Vert H^\prime(s)\Vert$ is uniformly bounded for $s\in [0,1]$. Furthermore,
we assume that the spectrum,  $\Sigma(s)$ of $H(s)$ can be decomposed in two parts:
$\Sigma(s) = \Sigma_1(s)\cup \Sigma_2(s)$, such that 
$\inf\{ \vert\lambda_1-\lambda_2\vert\ \mid \lambda_1\in\Sigma_1,  \lambda_2\in\Sigma_2\}= \gamma$ 
for a constant $\gamma>0$, uniformly in $s$.
We also assume there are compact intervals $I(s)$, with end points depending smoothly on
$s$ and such that $\Sigma_1(s)\subset I(s)\subset (\Rl\setminus\Sigma_2(s))$, in 
such a way that the distance between $I(s)$ and $\Sigma_2(s)$ has a strictly positive lower 
bound uniformly in $s$.
\end{assumption}

Typically, we have in mind a family of Hamiltonians of the form $H(s) = H(0) + \Phi(s)$,
with $H^\prime(s)=\Phi^\prime(s)$ bounded. Specifically, if $H(s)$ is unbounded, 
this is due to $H(0)$, which is obviously independent of $s$.
Let $E_{\lambda}(s)$ be the spectral family 
associated with $H(s)$ and let $P(s) := \int_{I(s)} dE_{\lambda}(s)$ be the spectral projection 
on the isolated part of the spectrum $\Sigma_1(s)$.

The formulation of the main result of this section uses a function $w_\gamma\in L^1(\bbR)$,
depending on a parameter $\gamma>0$, with the following properties.

\begin{assumption}\label{Assump:sgamma}
$w_\gamma\in L^1(\bbR)$ satisfies
\begin{enumerate}
\item $w_\gamma$ is real-valued and $\int dt\, w_\gamma(t) = 1$,
\item The Fourier transform $\widehat{w}_\gamma$ is supported in the interval
$[-\gamma,\gamma]$, i.e., $\widehat{w}_\gamma(\omega) = 0$, if $|\omega|\geq \gamma$.
\end{enumerate}
\end{assumption}

Such functions exist and were already considered in~\cite{hastings2010}.
In the following lemma, we present a family of such functions 
derived from ~\cite{I1934,dziubansky:1998}  and give explicit bounds
on their decay that we will need in this work and which may also 
prove useful in future applications.

\begin{lemma} \label{lemma:wdecay} 
Let $\gamma >0$ and define a positive sequence $(a_n)_{n\geq 1}$
by setting  $a_n= a_1  (n \ln^2n)^{-1}$ for $n\geq2$, and choosing $a_1$ so that 
$\sum_{n=1}^\infty a_n = \gamma/2$. 
Then, the infinite product 
\begin{equation}
\label{wDef}
w_\gamma(t) = c_\gamma\prod_{n=1}^\infty \Big(\frac{\sin a_n t}{a_nt}\Big)^2\,,
\end{equation}
defines an even, non-negative function  $w_\gamma\in L^1(\bbR)$, and we can choose 
$c_\gamma$ such that $\int w_\gamma(t)dt = 1$. With this choice, the following estimate holds.
For all $t\geq \mathrm{e}^{1/\sqrt{2}}/\gamma$,
\begin{equation}\label{eq:wdecay}
0 \leq w_\gamma(t) \leq 2 (\mathrm{e}\gamma)^2 t \cdot \exp\left(-\frac{2}{7}\frac{\gamma t}{\ln^2(\gamma t)}\right)\,.
\end{equation}

%\item For all $n\in\bbN$ and $t>\ep4/\gamma$,
%\begin{equation}
%\label{Moments}
%\int_t^\infty \xi^n w_\gamma(\xi) d\xi \leq \caP^{(2n+3)}((T(t)) \cdot\exp\left(-T(t)\right)\,,
%\end{equation}
%where $\caP^{(2n+3)}$ is a polynomial of degree $2n+3$, and 
%\begin{equation}
%T(t) = \frac{2}{7}\frac{\gamma t}{\ln^2(\gamma t)}
%\label{Toft}
%\end{equation}
%\end{enumerate}
\end{lemma}

\begin{proof} Without loss of generality, we shall assume $t\geq0$. Since each term of the product lies between $0$ and $1$, and by Stirling's formula,
$$
w_\gamma(t)\leq c_\gamma\prod_{n=1}^N \Big(\frac{\sin a_n t}{a_nt}\Big)^2 \leq c_\gamma(N!)^2 \ln^{4N}(N)(a_1t)^{-2N}
\leq 2\pi c_\gamma NN^{2N}\ln^{4N}(N)(a_1t)^{-2N}\ep{-2N}\,.
$$
The desired bound is obtained by choosing $N = \lfloor a_1t / \ln^2 (\gamma t)\rfloor$ and noting that $\gamma/7<a_1<\gamma/2$ 
and $\gamma/(2\pi)<c_\gamma<\gamma/\pi$. The bounds on $a_1$ follow directly, while the latter estimates are proven e.g. in \cite{BBB2008}. For $t>\mathrm{e}^{1/\sqrt{2}}/\gamma$, $N\leq\gamma t$ so that
$$
w_\gamma(t) \leq 2 (\mathrm{e}\gamma)^2 t \cdot \exp\left(-\frac{2}{7}\frac{\gamma t}{\ln^2(\gamma t)}\right)\,.
$$
Finally, this decay estimate and the a priori bound $w_\gamma(t)\leq 1$ for all $t$ imply that $w_\gamma\in L^1(\bbR)$.

\end{proof}

Since the Fourier transform of $\sin(ax) / (ax)$ is the indicator
function of the interval $[-a,a]$, the support of $\widehat{w}_\gamma$
corresponds to $[-2S,2S]$, where $S = \sum_{n=1}^\infty a_n$, and thus
(ii) of Assumption 2.2 also holds. Moreover, this lemma shows that the function $w_\gamma$ can be chosen
to decay faster than any power as $t\to\infty$. This will be important for some of 
our applications. We can now state and prove the main result of this section.

\begin{prop} \label{prop:main}
Let $H(s)$ be a family of self-adjoint operators satisfying  Assumption \ref{Assump:Hbd}.
Then, there is a norm-continuous family of unitaries $U(s)$ such that
the spectral projections $P(s)$ associated with the isolated portion of the spectrum 
$\Sigma_1(s)$, are given by
\begin{equation}
P(s) = U(s) P(0) U(s)\str.
\label{PUPU}\end{equation}
The unitaries are the unique solution of the linear differential equation
\begin{equation}
\label{QuasiAdabPr}
-\iu\frac{d}{d s} U(s) = D(s) U(s)\,,\qquad U(0) =\idty\,,
\end{equation}
where
\begin{equation}
\label{D}
D(s) = \int_{-\infty}^\infty dt\, w_\gamma (t) \int_0^t du\, \ep{\iu u H(s)} H^\prime(s)\ep{-\iu u H(s)}\,.
\end{equation}
for any function $w_\gamma$ satisfying Assumption \ref{Assump:sgamma}.
\end{prop}

It is obvious from \eq{D} and the assumption that $w_\gamma(t)\in\bbR$,
that  $D(s)$ is bounded, self adjoint, and the equations \eq{PUPU} and \eq{QuasiAdabPr}
can be combined into
\begin{equation}
\frac{d}{ds} P(s) = \iu [D(s), P(s)]\,.
\end{equation}
Moreover, boundedness of $D(s)$ implies that the unitaries $U(s)$ are norm
continuous.

The existence of a (bounded holomorphic) transformation function $V(s)$ such that $P(s) = 
V(s) P(0)V(s)^{-1}$ is a direct consequence of the smoothness of $P(s)$, see e.g.\ \cite{Kato}.
The interest of the proposition stems from having an explicit formula of a unitary family 
$U(s)$, from which interesting properties can be derived. This constructive aspect is
essential for the applications we have in mind (see Sections~\ref{Sec:LocalPert} 
and~\ref{Sec:LRforD}).
\begin{proof}
On the one hand,
\begin{equation}
P(s)= -\frac{1}{2\pi\iu}\int_{\Gamma(s)} dz\, R(z,s)\,,
\end{equation}
where $R(z,s) = (H(s) - z)^{-1}$ is the resolvent of $H(s)$ at $z$, and the contour $\Gamma(s)$ encircles the real interval $I(s)$ in the complex plane. Therefore,
\begin{equation}
P^\prime(s) = -\frac{1}{2\pi\iu}\int_{\Gamma(s)} dz\, R^\prime(z,s) = \frac{1}{2\pi\iu}\int_{\Gamma(s)} dz\, R(z,s)H^\prime(s)R(z,s)\,,
\end{equation}
where the first equality follows by noting that the smooth dependence of $s\mapsto I(s)$ and the uniform lower bound on the gap imply that the contour $\Gamma(s)$ can be kept fixed while differentiating; Namely for $\varepsilon$ small enough, $\Gamma(s)$ can be chosen so that it encircles all intervals $I(\sigma)$, $\sigma\in[s,s+\varepsilon]$. The $s$ dependence of $\Gamma$ can therefore be taken as purely parametric. Since $P(s)$ is an orthogonal projection, $P(s)P^\prime(s)P(s) = (1-P(s))P^\prime(s)(1-P(s)) = 0$ and therefore,
\begin{align}
P^\prime(s) &= \frac{1}{2\pi\iu}\int_{\Gamma(s)} dz\, \big( P(s) R(z,s)H^\prime(s) R(z,s)(1-P(s)) + (1-P(s)) R(z,s)H^\prime(s)R(z,s)P(s) \big) \nonumber \\
&= \frac{1}{2\pi\iu}\int_{\Gamma(s)} dz\! \int_{I(s)} d\mu \! \int_{\bbR/I(s)} d\lambda \, \frac{1}{\mu-z}\frac{1}{\lambda-z}\big( dE_{\mu}(s) H^\prime(s) dE_{\lambda}(s) + dE_{\lambda}(s) H^\prime(s) dE_{\mu}(s) \big) \label{DoubleOpL} \\
&= -\int_{I(s)} d\mu \! \int_{\bbR/I(s)} d\lambda \, \frac{1}{\lambda-\mu} \big( dE_{\mu}(s) H^\prime(s) dE_{\lambda}(s) + dE_{\lambda}(s) H^\prime(s) dE_{\mu}(s) \big)\,.
\label{final}
\end{align}
In order to justify the last equality, we interpret the double spectral integral
as a double operator integral, see e.g.\ \cite{BS2003}, Theorem 4.1{\it(iii)}. Eq.~(\ref{DoubleOpL}) corresponds to the factorization of the symbol $\phi(\lambda,\mu) = (\lambda-\mu)^{-1}$ of~(\ref{final}), the auxiliary measure space being $(S^1, d\gamma(t))$ where $S^1\ni t \mapsto \gamma(t)\in\bbC$ is a parametrization of $\Gamma(s)$. The uniform integrability conditions are met because of the finite size of the gap. On the other hand,
\begin{align}
\iu &[D(s), P(s)] = \iu\big( (1-P(s))D(s) P(s) - P(s) D(s) (1-P(s)) \big) \nonumber \\
&= \iu \int_{I(s)} d\mu \! \int_{\bbR/I(s)} d\lambda \! \int dt\, w_\gamma(t) \int_0^t du\big( \ep{\iu u (\lambda - \mu)}dE_{\lambda}(s) H^\prime(s) dE_{\mu}(s) - \ep{-\iu u (\lambda - \mu)}dE_{\mu}(s) H^\prime(s) dE_{\lambda}(s)\big)
\end{align}
which yields~(\ref{final}) after the time integrations are performed, namely
\begin{align*}
\iu \int dt\, w_\gamma(t) \int_0^t du\, \ep{\pm \iu u (\lambda - \mu)} &= \pm\int dt\, w_\gamma(t) \frac{1}{\lambda-\mu}\big(\ep{\pm \iu t (\lambda - \mu)}-1\big) \\
&= \pm \frac{1}{\lambda-\mu}\big(\widehat{w_\gamma}(\pm(\mu-\lambda))-1\big) = \mp \frac{1}{\lambda-\mu}\,,
\end{align*}
where we used first that $\int w_\gamma(t)=1$ and then the compact support property of $\widehat{w_\gamma}$ together with the fact that $|\lambda-\mu| > \gamma$ by Assumption~(\ref{Assump:Hbd}).
\end{proof}

We now introduce the weight function
\begin{equation}
\label{DefF}
W_\gamma(t) := \begin{cases}
\int_t^\infty d\xi\, w_\gamma(\xi) & t \geq 0 \\
- \int_{-\infty}^t d\xi\, w_\gamma(\xi) & t<0 
\end{cases}
\end{equation}
which will play a central role in the following applications. As $w_\gamma\in L^1(\bbR)$, $W_\gamma$ is well-defined. 
\begin{lemma}\label{lem:ua}
For $a>0$ define 
$$
u_a(\eta)=\ep{-a\frac{\eta}{\ln^2\eta}}\,,
$$
on the domain $\eta>1$. For all integers $k\geq 0$ and for all $t\geq \ep{4}$ such that also
$$
a\frac{t}{\ln^2 t}\geq 2k+2\,,
$$
we have the bound
$$
\int_{t}^\infty \eta^k u_a(\eta)\, d\eta\leq \frac{(2k+3)}{a}t^{2k+2}u_a(t)\, .
$$
\end{lemma}
\begin{proof}
For $\eta\geq \ep{2}$, the function 
$$
\tau(\eta) =a\frac{\eta}{\ln^2\eta}
$$ 
is positive, differentiable, and monotone increasing, and 
$$
\frac{d\eta}{d\tau}= \frac{1}{a}\left(\frac{\ln^2 \eta}{1-\frac{2}{\ln(\eta)}}\right) \leq \frac{\eta}{a}
$$
If we further require $\eta\geq \ep{4}$, we can also use the bound $1\leq \eta/(\log \eta)^4$,
and therefore
$$
\eta \leq \left(\frac{\eta}{\ln^2\eta}\right)^2 = \frac{\tau^2}{a^2}
$$
By making the substitution to the integration variable $\tau$ in the integral, we find
$$
\int_{t}^\infty \eta^k u_a(\eta)\, d\eta\leq\frac{1}{a^{2k+3}} \, \Gamma(2k+3,\tau(t))\,,
$$
where the incomplete Gamma function $\Gamma(n+1,x)$ can be computed for any integer $n\geq 0$ by repeated integration by parts:
$$
\Gamma(n+1,x) = \int_x^\infty \tau^n \ep{-\tau}\, d\tau = n! \, \ep{-x} \sum_{k=0}^n\frac{x^k}{k!}\,.
$$
For $x\geq n$, this yields the bound
$$
\Gamma(n+1,x) \leq (n+1)x^n\ep{-x}\,,
$$
which can be applied with $n=2k+2$ and $x = \tau(t)\leq a t$ to conclude the proof.
\end{proof}

\begin{lemma} \label{lemma:Bigwdecay} 
Let $\gamma>0$ and $w_\gamma$ the function defined in~(\ref{wDef}). Then eq.~(\ref{DefF}) defines a bounded, odd function $W_\gamma\in L^1(\bbR)$  with the following properties:
\begin{enumerate}
\item $| W_\gamma(t) |$ is continuous and monotone decreasing for $t\geq 0$.
In particular
\begin{equation}
\label{Wapriori}
\| W_\gamma\|_\infty = W_\gamma(0) = 1/2\,;
\end{equation}
\item  $| W_\gamma(t) |\leq G^{(W)}(\gamma | t |)$, with $G^{(W)}(\eta)$ defined for $\eta\geq 0$
by
\begin{equation} \label{eq:Geta}
G^{(W)}(\eta) = \begin{cases}
\frac{1}{2} & 0\leq \eta \leq \eta^*\\
35 \ep{2} \eta^4u_{ 2/7}( \eta) & \eta > \eta^*
\end{cases}
\end{equation}
where $\eta^*$ is the largest real solution of
$$
35 \ep{2} \eta^4u_{2/7}(\eta)=1/2\,.
$$
\item There is a constant $K$ such that
\begin{equation} \label{W1bound}
\|W_\gamma\|_1 \leq \frac{K}{\gamma}\,.
\end{equation}
\item For $t>0$, let
\begin{equation}\label{Igamma}
I_\gamma(t) = \int_t^\infty d\xi\, W_\gamma (\xi)\,.
\end{equation}
Then, $\vert I_\gamma(t) \vert \leq G^{(I)}(\gamma|t|)$, where $G^{(I)}(\zeta)$ is defined for $\zeta\geq 0$ by
$$
G^{(I)}(\zeta) = \frac{1}{\gamma} \cdot \begin{cases}
\frac{K}{2} & 0\leq \zeta \leq \zeta ^* \\
130\ep{2}\zeta^{10} u_{2/7}(\zeta) & \zeta > \zeta ^*
\end{cases}\,.
$$
with $K$ as in (iii) and a $\zeta ^*>0$.
\end{enumerate}
\end{lemma}
\begin{rem}
It is straightforward to estimate the values of the constants $\eta^*$, $\zeta^*$,
and $K$, by numerical integration. One finds $14250<\eta^*<14251$,
$36057<\zeta^*<36058$, and $K\sim14708$.
\end{rem}
\begin{proof} 

i. $w_\gamma\geq0$, even, and  $\int w_\gamma = 1$. With the definition~(\ref{DefF}) of $W_\gamma$, this implies 
\begin{equation}\label{Wpointwise}
| W_\gamma(t)|\leq \int_{|t|}^\infty w_\gamma(\xi)\, d\xi
\leq \int_0^\infty w_\gamma(\xi)\, d\xi = W_\gamma(0) = \frac{1}{2}.
\end{equation}

ii. The bound (\ref{eq:wdecay}) for $w_\gamma$ gives
$$
| W_\gamma(t) |= \int_{|t|}^\infty d\xi\, w_\gamma(\xi) \leq 2\ep{2} \gamma^2 \int_{|t|}^\infty d\xi\, \xi u_{2/7}(\gamma\xi) = 2\ep{2} \int_{\gamma|t|}^\infty d\eta\, \eta u_{2/7}(\eta)\,.
$$
With $k=1$ and $a=2/7$, the conditions of Lemma \ref{lem:ua} are satisfied
for $\gamma |t|\geq 561$, so that
\begin{equation}\label{800bound}
| W_\gamma(t)|\leq 35 \ep{2} (\gamma | t|)^4 u_{2/7}(\gamma |t|)\,,\qquad\text{ if }\gamma |t|\geq 561.
\end{equation}
Using the decay of $u_a(\eta)$ for $\eta\geq \ep{2}$ and the fact that the RHS of (\ref{800bound}) exceeds the a priori bound~(\ref{Wapriori}) for $\gamma |t| = 561$, the result follows. 

iii. By (ii) $W_1\in L^1(\mathbb{R})$ and $|W_\gamma(t)| \leq |W_1(\gamma t)|$, which implies the existence of a constant $K$ as claimed. 
Using the oddness of $W_\gamma$ and the explicit function $G^{(W)}(\eta)$, we choose
$$
K = \eta^* +70 \ep{2}  \int_{\eta^*}^\infty \eta^4 u_{2/7}(\eta)\, d\eta.
$$

iv. Follows by (iii) and another application of Lemma \ref{lem:ua}.

\end{proof}

A straightforward corollary of the decay conditions of the weight function is the following equivalent form of the generator $D(s)$, eq.~(\ref{D}).
\begin{cor}
\label{cor:DW}
The conclusions of Proposition~(\ref{prop:main}) hold with
\begin{equation}
D(s) = \int_{-\infty}^\infty dt\, W_\gamma(t) \cdot \ep{\iu t H(s)} H^\prime(s) \ep{-\iu t H(s)}\,.
\label{eq:D_S}
\end{equation}
with $W_\gamma$ as in lemma~\ref{lemma:Bigwdecay}.
\end{cor}
\proof This follows by a simple integration by parts from~(\ref{D}). By definition of the function $W_\gamma$, we have, for any $t\in\bbR\setminus\{0\}$, 
\begin{equation*}
\frac{d}{dt} W_\gamma(t) = -w_\gamma(t)\,,
\end{equation*}
which can be extended by continuity at $t=0$. Proposition~(\ref{prop:main}) then yields
\begin{equation*}
D(s) = -\left. W_\gamma(t) \int_0^t du\, \ep{\iu u H(s)} H^\prime(s) \ep{-\iu u H(s)}\right|_{-\infty}^\infty + \int_{-\infty}^\infty dt\, W_\gamma(t) \cdot \ep{\iu t H(s)} H^\prime(s) \ep{-\iu t H(s)}\,.
\end{equation*}
The boundary term vanishes by Assumption~(\ref{Assump:Hbd}) and the decay of $W_\gamma$.\teop

%%%%%%%%%%%%%%%%%%%%%%%%%%%%%%%%%%%%%%%%%%%%%%%%%%%%%%%%%%%%%%%%%%%%%%%%%%%%%%%%%%%%%%%%%%%%%%%%%%%%%%%%

\section{Local perturbations}
\label{Sec:LocalPert}

The aim of this section is to combine the evolution formula of Section~\ref{sec:Evolution}
with Lieb-Robinson bounds to show that the effect of perturbations 
with a finite support $X$ can be, to arbitrarily good approximation, expressed by
the action of a local operator with a support that is a moderate enlargement of $X$.
In principle, the following lemma suffices to turn Lieb-Robinson bounds into an 
estimate for the support of a time-evolved observable. 

\begin{lemma}[\cite{nw2011}]
\label{lma:Approx_non-normal}
Let $\cH_1$ and $\cH_2$ be Hilbert spaces and suppose $\epsilon \geq 0$
and $A\in\cB(\cH_1\otimes\cH_2)$ are such that
$$
\Vert [A,\idty\otimes B]\Vert \leq \epsilon \Vert B\Vert\, \mbox{ for all } B\in\cB(\cH_2).
$$
Then, there exists $A^\prime\in\cB(\cH_1)$, such that
\be
\Vert A^\prime\otimes \idty -A\Vert\leq \epsilon.
\label{best_small_commutator}\ee
\end{lemma}
If $\dim \cH_2 <\infty$, one can simply take 
$$
A^\prime=\frac{1}{\dim \cH_2}\Tr_{\cH_2} A ,
$$
as is done in \cite{bravyi:2006,nach22006} (or see (i) in the proof of
Lemma~\ref{lma:Approx} below).

For the applications we have in mind, we want the map $A\mapsto A^\prime$
to be continuous in the weak operator topology. In finite dimensions the partial
trace is of course continuous. In infinite dimensions we cannot use the partial
trace and the continuity is not obvious. Moreover, it will be convenient for us
to have a map $A^\prime=\Pi(A)$ that is compatible with the tensor product 
structure of the algebra of local observables of a lattice system (see 
Section~\ref{sec:lrbprop}). For this purpose, we fix a normal state $\rho$ on $\cB(\cH_2)$ 
and define the map 
$\Pi: \cB(\cH_1)\otimes\cB(\cH_2)\to \cB(\cH_1)
\cong \cB(\cH_1)\otimes\idty\subset\cB(\cH_1) \otimes\cB(\cH_2)$ by $\Pi=\id\otimes \rho$.
Although the map $\Pi$ depends on $\rho$, we have the following estimate
independent of $\rho$.

\begin{lemma}
\label{lma:Approx}
Let $\cH_1$ and $\cH_2$ be Hilbert spaces and suppose $\epsilon \geq 0$
and $A\in\cB(\cH_1\otimes\cH_2)$ are such that
$$
\Vert [A,\idty\otimes B]\Vert \leq \epsilon \Vert B\Vert\, \mbox{ for all } B\in\cB(\cH_2).
$$
Then,
\be
\Vert \Pi(A) -A\Vert\leq 2\epsilon.
\ee
\end{lemma}

\begin{proof}
(i) First, assume $\dim \cH_2< \infty$. Then it suffices to take for $A^\prime$
the normalized partial trace of $A$:
$$
A^\prime=\frac{1}{\dim \cH_2}\Tr_{\cH_2} A
$$
Note that
$$
A^\prime\otimes \idty = \int_{\cU(\cH_2)} dU \, (\idty\otimes U^*)A(\idty\otimes U)
$$
where $dU$ is the Haar measure on the unitary group, $\cU(\cH_2)$, of $\cH_2$.
Then, by the assumptions of the Lemma, one has
$$
\Vert A^\prime\otimes \idty -A\Vert 
\leq \int_{\cU(\cH_2)} dU \, \Vert(\idty\otimes U^*)[A,(\idty\otimes U)]\Vert \leq \epsilon\, .
$$

(ii) In the case of infinite-dimensional $\cH_2$, we start by defining, for $\eta\in\cH_2$,
$\Vert\eta\Vert =1 $, $A_\eta\in\cB(\cH_1)$ by the formula
$$
\langle \phi, A_\eta \psi \rangle = \langle \phi\otimes\eta, A \psi\otimes \eta \rangle,
\quad \phi,\psi \in \cH_1\, .
$$
For $\eta,\xi\in\cH_2$, let $\kettbra{\xi}{\eta}$ denote the 
rank-1 operator defined by $\kettbra{\xi}{\eta}\phi=\langle\eta,\phi\rangle\xi$,
for all $\phi\in\cH_2$. For any three $\eta,\xi,\chi\in\cH_2$, $\Vert\eta\Vert=\Vert \xi\Vert
=\Vert\chi\Vert=1$, note that
\be
A_\xi\otimes \kettbra{\eta}{\chi} = (\idty\otimes\kettbra{\eta}{\xi})A 
(\idty\otimes\kettbra{\xi}{\chi}) \, .
\label{magic}\ee
This equation is easily verified by equating matrix elements with arbitrary tensor
product vectors $\phi\otimes\alpha$ and $\psi\otimes\beta$.
By the assumptions we then have 
$$
\Vert (\idty\otimes\kettbra{\eta}{\xi})
\left[ A,\idty\otimes\kettbra{\xi}{\eta}\right](\idty\otimes\kettbra{\eta}{\xi})\Vert
\leq \Vert \left[ A,\idty\otimes\kettbra{\xi}{\eta}\right]\Vert
\leq \epsilon \, .
$$
By expanding the commutator and simplifying the products 
in the left hand side of this inequality and using (\ref{magic}) we obtain 
\be
\Vert A_\xi-A_\eta\Vert=\Vert A_\xi\otimes\kettbra{\eta}{\xi}
 -A_\eta\otimes\kettbra{\eta}{\xi}\Vert\leq \epsilon \, .
\label{diam}\ee
Next, consider finite-dimensional orthogonal projections $P$ on $\cH_2$. 
Since, for each such $P$, 
$$
\Vert [(\idty\otimes P) A(\idty\otimes P), \idty\otimes (PBP)]\Vert
=\Vert  [(\idty\otimes P) [A,PBP](\idty\otimes P)\Vert
\leq \Vert [A, PBP]\Vert\leq \epsilon \Vert B\Vert,
$$
by (i), there exists $A_P\in \cB(\cH_1)$ such that
\be
\Vert A_P\otimes P - (\idty\otimes P)A(\idty\otimes P)\Vert \leq \epsilon\,.
\label{PQ}\ee
Explicitly, if $\chi_1,\ldots,\chi_n$ is an o.n. basis of $\ran\ P$, the construction in part (i) provides
$$
A_P=\frac{1}{n}\sum_{k=1}^n A_{\chi_k}, \quad \mbox{and }\Vert A_P\Vert \leq \Vert A\Vert\,.
$$
The diameter of the convex hull of $\{ A_\chi \mid \chi\in\cH_2, \Vert \chi\Vert =1\}$ is
bounded by $\epsilon$ due to (\ref{diam}). It follows that for any two finite-dimensonial
projections $P,Q$ on $\cH_2$
$$
\Vert A_P - A_Q\Vert \leq \epsilon\,.
$$
Now, we prove the bound:
$$
\Vert A_P\otimes \idty -A\Vert \leq 2\epsilon
$$
by contradiction. Suppose that for some $P$, 
$\Vert A_P\otimes \idty -A\Vert > 2\epsilon$. Then, there exists $\delta >0$
such that $\Vert A_P\otimes \idty -A\Vert > 2\epsilon+\delta$. Therefore, there exist
$\phi,\psi\in \cH_\Lambda$, $\Vert \phi\Vert = \Vert \psi\Vert =1$, such that
$$
\vert\langle \phi ,(A_P\otimes \idty -A)\psi\rangle\vert > 2\epsilon+\frac{\delta}{2}\, .
$$
Let $Q$ be a finite-dimensional projection on $\cH_2$ such that
$$
\Vert (\idty-\idty\otimes Q)\phi\Vert \leq \frac{\delta}{8\Vert A\Vert}, \quad\mbox{and }
\Vert (\idty-\idty\otimes Q)\psi\Vert \leq \frac{\delta}{8\Vert A\Vert}\, .
$$
Then,
$$
\vert\langle \phi,(\idty\otimes Q)(A_P\otimes \idty)(\idty\otimes Q)\psi\rangle
- \langle\phi, (\idty\otimes Q)A(\idty\otimes Q)\psi\rangle\vert
> 2\epsilon +\frac{\delta}{2}-4\frac{\delta}{8\Vert A\Vert}\Vert A\Vert\, .
$$
Since $\Vert A_P-A_Q\Vert\leq \epsilon$, this implies
$$
\vert\langle\phi,(A_Q\otimes Q - (\idty\otimes Q)A(\idty\otimes Q))\psi\rangle\vert > \epsilon\,.
$$
which contradicts (\ref{PQ}).

To conclude the proof, note that for a density matrix in diagonal form, 
$\rho=\sum_k \rho_k \ketbra{\xi_k}{\xi_k}$, we have that $\id\otimes\rho(A)=\sum_k \rho_k
A_{\xi_k}$. Therefore we have
$$
\Vert \Pi(A)-A\Vert = \Vert \sum_k \rho_k A_{\xi_k}\otimes\idty-A\Vert
\leq \sum_k \rho_k \Vert  A_{\xi_k}\otimes\idty-A\Vert\leq \sum_k \rho_k 2\epsilon=2\epsilon.
$$
\end{proof}

We now explain a {\em local perturbations perturb locally} principle that applies in general
to any states corresponding to an isolated part of the spectrum of a system of which the 
dynamics has a quasi-locality property expressed by an estimate of Lieb-Robinson type. 
The basic argument can be applied for finite systems or for infinite systems in a suitable
representation. For the sake of presentation, we consider a systems defined on a 
metric graph $(\Gamma, d)$. To each site $x \in \Gamma$, we associate a Hilbert space 
$\mathcal{H}_x$. For finite $\Lambda \subset \Gamma$, we define
\begin{equation} \label{eq:hilale}
\mathcal{H}_{\Lambda} = \bigotimes_{x \in \Lambda} \mathcal{H}_x \quad \mbox{and} \quad
\mathcal{A}_{\Lambda} = \bigotimes \mathcal{B}( \mathcal{H}_x) 
\end{equation}
where $\mathcal{B}( \mathcal{H}_x)$ denotes the bounded linear operators over $\mathcal{H}_x$. 
There is a natural way to identify $\mathcal{A}_{\Lambda_0} \subset \mathcal{A}_{\Lambda}$; 
namely identify each $A \in \mathcal{A}_{\Lambda_0}$ with 
$A \otimes \idty_{\Lambda \setminus \Lambda_0} \in \mathcal{A}_{\Lambda}$.
We can then inductively define
\begin{equation} \label{eq:localg}
\mathcal{A}_{\rm loc} = \bigcup_{\Lambda \subset \Gamma} \mathcal{A}_{\Lambda}
\end{equation}
 where the union is taken over all finite subsets of $\Gamma$. The completion
 of $\mathcal{A}_{\rm loc}$ with respect to the operator norm is a $C^*$-algebra,
 which we will assume to be represented on a Hilbert space and assume that 
a family of Hamiltonians of the form $H(s) = H(0) + \Phi(s)$ on this space satisfies
Assumption~\ref{Assump:Hbd}. Additionally, we assume that the Heisenberg dynamics
$\tau_t^{H(s)}$, generated by $H(s)$, satisfies a Lieb-Robinson bound uniform in $s$.
\begin{assumption}
\label{Assump:LRBound}
There are constants $C(A,B)$, $a>0$ and a
Lieb-Robinson velocity $v\geq 0$ such that for all $s\in[0,1]$
$$
\Vert [\tau^{H(s)}_t(A),B]\Vert\leq C(A,B) \ep{- a \big(d(\supp A,\supp B)-v\vert t\vert\big)}
$$
Here, $C(A,B)$ is of a suitable form such as $C\Vert A\Vert\,\Vert B\Vert \min(\vert\supp A\vert,
\vert\supp B\vert)$.
\end{assumption}
Furthermore, we assume that there is a fixed finite subset $X\subset\Gamma$ such that $\Phi'(s)\in
\cA_{X}$ and
\begin{equation}
\| \Phi' \| = \sup_{0 \leq s \leq 1} \| \Phi'(s) \| < \infty.
\end{equation}

The generator $D(s)$ defined in \eq{D} and \eq{eq:D_S} for the local perturbation $\Phi(s)$ 
is not strictly local. However, the fast decay of the weight 
function $W_\gamma(t)$ in combination with Assumption~\ref{Assump:LRBound}
imply that the effect of $D(s)$ is small far away from $X$. To make this precise, let $R>0$, 
and denote by $X_R$ the following `fattening' of $X$:
\begin{equation}
\label{FatX}
X_R=\{x\::\:\exists y\in X \mbox{ s.t. } d(x,y)\leq R\}\,.
\end{equation}
The following result shows that in the situation described above the unitary  $U(s)$
of \eq{PUPU} in Proposition~\ref{prop:main} can be well-approximated by
a unitary $V_R(s)\in\cA_{X_R}$, i.e., with support in $X_R$.
\begin{thm}[Local Perturbations Perturb Locally]
\label{thm:lppl}
For any $R>0$, there exist unitary operators $V_R(s)$ with $\supp(V_R(s)) \subset X_R$ and a constant 
$C$, independent of $R$, such that
$$
\Vert U(s)-V_R(s)\Vert \leq C G^{(I)}(\frac{\gamma R}{2v})
$$
with $G^{(I)}$ the subexponential function defined in Lemma~\ref{lemma:Bigwdecay}. Consequently, we also have
\begin{equation}
\label{MainApprox}
\Vert P(1) - V_R(1) P(0) V_R(1)\str\Vert \leq 2C G^{(I)}(\frac{\gamma R}{2v}).
\end{equation}
\end{thm}

\begin{proof}
We begin by defining a local approximation of the self-adjoint generator
$D(s)$ starting from \eq{eq:D_S}. Consider the decomposition 
$\cA_{\rm loc}= \cA_{X_R}\otimes  \cA_{\Gamma\setminus X_R}$ and let $\Pi_R=\id\otimes
\rho$ for some state $\rho$ on $\cA_{\Gamma\setminus X_R}$, and define
$$
D_R(s) = \int_{-\infty}^\infty dt\, W_\gamma(t) \Pi_R( \ep{\iu t H(s)} \Phi^\prime(s) \ep{-\iu t H(s)})\,.
$$
Then, for any $T>0$ we have the following estimate:
$$
\Vert D(s)-D_R(s)\Vert
\leq \| \Phi ' \| \int_{\vert t\vert > T} \!dt\,\vert W_\gamma(t)\vert
  + \Vert W_\gamma\Vert_\infty
\int_{-T}^T \!dt\,  \Vert (\id -\Pi_R)( \ep{\iu t H(s)} \Phi^\prime(s) \ep{-\iu t H(s)})\Vert\,.
$$
For the first term, we apply the bound of Lemma~\ref{lemma:Bigwdecay} part (iv) and for the 
second term we use \eq{Wapriori} and Lemma~\ref{lma:Approx} and 
Assumption~\ref{Assump:LRBound} to get
$$
\Vert D(s)-D_R(s)\Vert\leq \Vert \Phi^\prime\Vert 2C G^{(I)}(\gamma T)
+ \frac{1}{2} C\Vert \Phi^\prime\Vert \vert X\vert e^{- a(R-vT)}.
$$
For the simple choice $T=R/(2v)$, for not too small $R$, the second term is 
negligible compared to the first , and we obtain
\begin{equation}
\Vert D(s)-D_R(s)\Vert\leq C^\prime \Vert \Phi^\prime\Vert G^{(I)}(\frac{\gamma R}{2v}).
\label{local_D}\end{equation}
Now, let $V_R(s)$ be solution of
$$
-i\frac{\partial}{\partial s} V_R(s)= D_R(s) V_R(s)\,,\qquad V_R(0)=\idty.
$$
The claim follows by integrating the estimate \eq{local_D}.
\end{proof}

To illustrate this result, we consider the case where the isolated part of the spectrum, 
$\Sigma_1(s)$ in Assumption 2.1, consists
of a non-degenerate ground state energy. Let $\psi_0(s)$ denote the
corresponding normalized eigenvector and let $A\in\cA_{\Lambda\setminus X_R}$ 
be an observable supported away from the perturbation, whence $[A, V_R] = 0$. 
By applying Theorem~\ref{thm:lppl} we immediately obtain
\begin{eqnarray*}
\vert \langle \psi(s), A\psi(s)\rangle - \langle \psi(0), A\psi(0)\rangle\vert
&=&  \vert \langle \psi(0), U(s)^*[A, U(s)]\psi(0)\rangle\vert\\
&=&  \vert \langle \psi(0), U(s)^*[A, U(s)-V_R(s)]\psi(0)\rangle\vert\\
&\leq& 2\Vert A\Vert \Vert U(s)-V_R(s)\Vert \leq 2C \Vert A\Vert
G^{(I)}(\frac{\gamma R}{2v})
\end{eqnarray*}
This estimate clearly expresses the locality of the effect of the perturbation on the state $\psi(s)$.

%%%%%%%%%%%%%%%%%%%%%%%%%%%%%%%%%%%%%%%%%%%%%%%%%%%%%%%%%%%%%%%%%%%%%%%%%%%%%%%%%%%%%%%%%%%%%%%%%%%%%%%%

\section{The spectral flow and quasi-locality}
\label{Sec:LRforD}

The main goal of this section is to prove that the spectral flow defined in terms of the
unitary operators $U(s)$, as in Proposition~\ref{prop:main}, satisfies a Lieb-Robinson bound. 
This is the content of Theorem~\ref{thm:lrbpe} below. In Section~\ref{sec:lrbprop}, we introduce
the basic models to which our result applies and state Theorem~\ref{thm:lrbpe}. Our
proof of Theorem~\ref{thm:lrbpe} demonstrates that the claimed estimate 
follows from a Lieb-Robinson bound for time-dependent interactions. We state and
prove a general result of this type, see Theorem~\ref{thm:lrbtdints}, in Section~\ref{sec:lrbtdints}. 
The remainder of Section~\ref{Sec:LRforD}
is used  prove that Theorem~\ref{thm:lrbtdints} is applicable in the context of the spectral flow.
Section~\ref{sec:lemma} contains a technical lemma, and Section~\ref{sec:proof} finishes
the proof.

\subsection{The set-up and a statement of the main result} \label{sec:lrbprop}

The arguments we provide in Section~\ref{Sec:LRforD} apply to a
large class of models. In this subsection, we describe in detail the
assumptions necessary to prove a Lieb-Robinson bound
for the spectral flow.

We will consider models defined on a countable set $\Gamma$ equipped with a metric $d$.
Typically, $\Gamma$ will be infinite, e.g., $\Gamma = \mathbb{Z}^{\nu}$.
In the case that $\Gamma$ is infinite, we require some assumptions
on the structure of $\Gamma$ as a set. First, we will assume a uniform bound on the
rate at which balls grow, i.e., we assume there exist numbers $\kappa >0$ and
$\nu >0$ for which 
\begin{equation} \label{eq:ballgrow}
\sup_{x \in \Gamma} |B_r(x) | \leq \kappa r^{\nu} \, ,
\end{equation}
where $|B_r(x)|$ is the cardinality of the ball centered at $x$ of radius $r$.
In addition, we will assume that $\Gamma$ has some 'integrable' underlying 
structure. We express this property in terms of a non-increasing, real-valued 
function $F: [0, \infty) \to (0, \infty)$ that satisfies
\newline i) {\it uniform integrablility}: i.e. 
\begin{equation} \label{eq:uniint}
\| F \| = \sup_{x \in \Gamma} \sum_{y \in \Gamma} F(d(x,y)) \, < \, \infty \, 
\end{equation}
and
\newline ii) {\it a convolution condition}: i.e., there exists a number $C_F>0$ such that
given any pair $x,y \in \Gamma$, 
\begin{equation} \label{eq:convc}
\sum_{z \in \Gamma} F(d(x,z)) F(d(z,y) \leq C_F F(d(x,y)) \, .
\end{equation} 
For the case of $\Gamma = \mathbb{Z}^{\nu}$, one possible choice of $F$ is given by 
$F(r) = (1+r)^{-(\nu + 1)}$. The corresponding convolution constant 
may be taken as $C_F = 2^{\nu+1} \sum_{x \in \Gamma}F(|x|)$. 

Lastly, we need an assumption on the rate at which $F$ goes to zero. It is
convenient to express this in terms of the sub-exponential function $u_a$ introduced
in Lemma~\ref{lem:ua}. We suppose that there exists a number $0< \delta < 2/7$ such that
\begin{equation} \label{eq:Fgrowth}
\sup_{r \geq 1} \frac{u_{\delta}(r)}{F(r)} < \infty \, .
\end{equation}
Clearly, if $\Gamma = \mathbb{Z}^{\nu}$ and $F(r) =(1+r)^{-(\nu+1)}$, then (\ref{eq:Fgrowth}) holds
for every $0 < \delta <2/7$.

The following observations will be useful. 
Let $F:[0, \infty) \to (0, \infty)$ be a non-increasing function satisfying \eq{eq:uniint}  and \eq{eq:convc}.
For each $a \geq 0$, the function $F_a(r) =e^{-ar} F(r)$
also satisfies the properties \eq{eq:uniint}  and \eq{eq:convc} with $\Vert F_a\Vert \leq \Vert F \Vert$ and $C_{F_a}\leq C_F$.
In fact, more generally, if  $g$ is positive, non-increasing, and logarithmically super-additive, 
i.e., $g(x+y)\geq g(x)g(y)$, then $F_g(r) = g(r) F(r)$ satisfies (\ref{eq:uniint}) and (\ref{eq:convc}) with $\|F_g\| \leq g(0) \|F\|$ 
and $C_{F_g} \leq C_F$. For brevity we will write $F_a$ to denote the case $g(r)=e^{-ar}$.
Other functions $g$ will be used later.

Recall the general quantum systems corresponding to $\Gamma$ on which our models
will be defined. As in Section~\ref{Sec:LocalPert}, we associate a Hilbert space $\mathcal{H}_{\Lambda}$
and an algebra of observables $\mathcal{A}_{\Lambda}$ to each finite set $\Lambda \subset \Gamma$,
see (\ref{eq:hilale}), and similarly define $\mathcal{A}_{\rm loc}$ as in (\ref{eq:localg}). 
In this case, the models we consider are comprised of two types of terms.
First, we fix a collection of Hamiltonians, which we label by
$\left( H_{\Lambda}(0) \right)_{\Lambda}$, with the property that for each finite $\Lambda \subset \Gamma$,
$H_{\Lambda}(0)$ is a densely defined, self-adjoint operator on $\mathcal{H}_{\Lambda}$.  
Next, we consider a family of interactions $\Phi(s)$ parametrized by a real number $s$.
For each $s$, the interaction $\Phi(s)$ on $\Gamma$ is a mapping from the set of
finite subsets of $\Gamma$ into $\mathcal{A}_{\rm loc}$ with the
property that $\Phi(X,s)^* = \Phi(X,s) \in \mathcal{A}_X$ for all
finite $X \subset \Gamma$. It is convenient to write $\Phi(X,s) = \Phi_X(s)$. 
A model then consists of a choice of
$\left( H_{\Lambda}(0) \right)_{\Lambda}$ and a family of interactions $\Phi(s)$ over $\Gamma$.
Given a model, we associate local Hamiltonians to each finite set 
$\Lambda \subset \Gamma$ by setting
\begin{equation} \label{eq:locham}
H_{\Lambda}(s) = H_{\Lambda}(0) + \sum_{X \subset \Lambda} \Phi_X(s)
\end{equation}
where the sum is taken over all subsets $X \subset \Lambda$. For notational
consistency, we will assume that $\Phi_X(0) = 0$ for all $X$.
With $s$ fixed, the sum in (\ref{eq:locham}) above is finite for each such $\Lambda \subset \Gamma$, and
thus self-adjointness guarantees the existence of the Heisenberg dynamics, i.e., 
\begin{equation}
\tau_t^{H_{\Lambda}(s)}(A) = e^{itH_{\Lambda}(s)} A e^{-itH_{\Lambda}(s)}  \quad \mbox{for all } A \in \mathcal{A}_{\Lambda} \, \mbox{ and } t \in \mathbb{R},
\end{equation}
which, again for fixed $s$, is a one-parameter group of automorphisms on $\mathcal{A}_{\Lambda}$. 

To prove the results in this section, we need a boundedness assumption on the family of
interactions. We make this precise by introducing a norm on the interactions $\Phi(s)$ over $\Gamma$, with
respect to any non-increasing, positive function $F$ satisfying \eq{eq:uniint}  and \eq{eq:convc}, as follows:
\begin{equation}\label{Fnorm}
\| \Phi \|_F = \sup_{x,y \in \Gamma} \frac{1}{F(d(x,y))} 
\sum_{\stackrel{Z \subset \Gamma:}{x,y \in Z}} \sup_{s} \| \Phi_Z(s) \| \, < \infty \, .
\end{equation}
The sum above is over all finite sets $Z \subset \Gamma$ containing $x$ and $y$, and we will often abbreviate $\Vert \cdot\Vert_{F_a}$ by  $\Vert \cdot\Vert_a$. On occasion,
we will use $\| \Phi(s) \|_F$ for the norm of $\Phi(s)$ at fixed $s$, i.e., the norm
defined by dropping the supremum over $s$ in \eq{Fnorm}.
The following lemma states some simple bounds in terms of $\Vert \Phi\Vert_F$ that 
we will frequently use.
\begin{lemma}\label{lem:simple_bounds}
Let $\Phi(s)$ be a family of interactions over $\Gamma$ for which $\Vert \Phi\Vert_F <\infty$ for
some non-increasing, positive function $F$ satisfying \eq{eq:uniint}  and \eq{eq:convc} above.
Then, for any finite $\Lambda \subset \Gamma$, we have
\bea
\sum_{\substack{X \subset \Lambda: \\ x\in X}}\Vert \Phi_X(s)\Vert&\leq &F(0)\Vert\Phi\Vert_F\\
\sum_{X\subset\Lambda}\Vert\Phi_X(s)\Vert&\leq& F(0)\Vert \Phi\Vert_F \vert\Lambda\vert\,.
\eea
\end{lemma}
\begin{proof}
For $x\in \Gamma$ we have
$$
\sum_{\substack{X \subset \Lambda: \\ x\in X}}\Vert \Phi_X(s)\Vert
\leq\sup_{y\in \Gamma} F(d(x,y))
\sum_{\stackrel{X \subset \Gamma:}{x,y \in X}} \frac{\Vert \Phi_X(s)\Vert}{F(d(x,y))}\\
\leq F(0)\Vert\Phi\Vert_F
$$
where we have used the definition of the norm \eq{Fnorm} and the monotonicity of $F$.
Using this estimate, for any finite subset $\Lambda\subset \Gamma$, we obtain the bound
$$
\sum_{X\subset\Lambda}\Vert\Phi_X(s)\Vert\leq\sum_{x\in\Lambda}
\sum_{\substack{X \subset \Lambda: \\ x\in X}}\Vert \Phi_X(s)\Vert\leq F(0)\Vert \Phi\Vert_F \vert\Lambda\vert\,.
$$
\end{proof}

We will also require the interactions to be smooth with bounded derivatives. More concretely,
let $\Phi(s)$ be a family of interactions over $\Gamma$ for which, given any finite $X \subset \Gamma$, 
$\Phi_X(s)$ is differentiable with respect to $s$. In this case, we define a corresponding family of 
interaction $\partial\Phi(s)$ over $\Gamma$ by the the formula
$$
\partial\Phi_X(s)=\vert X\vert \Phi_X'(s)\,  \quad \mbox{for each finite } X \subset \Gamma.
$$

We now state the main assumptions of this section.

\begin{assumption} \label{ass:dphi}
We will assume that the interactions $\Phi(s)$ are differentiable with respect to $s$. More specifically, 
we assume that for each finite $X \subset \Gamma$, $\Phi_X'(s) \in \mathcal{A}_X$ for all $s$.
In addition, we suppose a uniform estimate on the norms of these derivatives as $s$ varies in
compact sets. For concreteness, we will assume that the domain of $s$-values is  $[0,1]$, 
and suppose that there exists a number $a>0$ for which
$$
\Vert \partial\Phi\Vert_a <\infty.
$$
\end{assumption}

\begin{assumption} \label{ass:unigap}
We will assume that for every finite $\Lambda \subset \Gamma$, the local Hamiltonian $H_{\Lambda}(s)$ 
has a spectrum that is uniformly gapped. More precisely, the spectrum of $H_{\Lambda}(s)$, which we will denote
by $\Sigma^{(\Lambda)}(s)$, can be decomposed into two non-empty sets: $\Sigma^{(\Lambda)}(s) = \Sigma^{(\Lambda)}_1(s) \cup \Sigma^{(\Lambda)}_2(s)$  
with $d( \Sigma^{(\Lambda)}_1(s), \Sigma^{(\Lambda)}_2(s) ) \geq \gamma >0 $. In particular, the positive number $\gamma$
is independent of $s \in [0,1]$ and finite $\Lambda \subset \Gamma$. We also suppose that there exist
intervals $I(s)$, with endpoints depending smoothly on $s$, for which $\Sigma^{(\Lambda)}_1(s) \subset I(s)$. 
\end{assumption}
In typical applications, the set $\Sigma^{(\Lambda)}_1(s)$ will consist of the ground state and (possibly) 
other low-lying energies, but this is not necessary.

Given Assumptions~\ref{ass:dphi} and \ref{ass:unigap}, the results of Section~\ref{sec:Evolution} apply to the local
Hamiltonians $H_{\Lambda}(s)$. We need a further assumption in order to state the main result of this
section.

\begin{assumption} \label{ass:unilrb}
We will assume a uniform, exponential Lieb-Robinson bound. In fact, we assume that there exists an $a>0$ and
numbers $K_a$ and $v_a$ such that 
\begin{equation} \label{eq:unilrb}
\left\| \left[ \tau_t^{H_{\Lambda}(s)}(A), B \right] \right\| \leq K_a \| A \| \| B \| e^{a v_a |t|} \sum_{x \in X, y \in Y} F_a(d(x,y))
\end{equation}
holds for all $A \in \mathcal{A}_X$, $B \in \mathcal{A}_Y$, and $t \in \mathbb{R}$. Here, as above, $F_a(r) =e^{-ar}F(r)$,
and we stress that the numbers $K_a$ and $v_a$ are each independent of both $\Lambda$ and $s$. 
\end{assumption}
Estimates of the form (\ref{eq:unilrb}) have been demonstrated for a number of models, see e.g. \cite{nach2010},
and references therein, for a recent review. Here we assume it holds for a class of models, and as a consequence, 
we get Theorem~\ref{thm:lrbpe} below.

As indicated above, given Assumptions~\ref{ass:dphi} and \ref{ass:unigap}, the results of Proposition~\ref{prop:main}
apply to $H_{\Lambda}(s)$ for each finite $\Lambda \subset \Gamma$ and $s \in [0,1]$. 
In this case, there are unitaries $U_{\Lambda}(s)$ in terms of which we define the
following {\it spectral flow}:
\begin{equation} \label{automorphism}
\alpha_s^{\Lambda}(A) = U_{\Lambda}(s)^* A U_{\Lambda}(s) \quad \mbox{for all } A \in \mathcal{A}_{\Lambda} \quad \mbox{and } 0 \leq s \leq 1 \, .
\end{equation}

The main result of this section is a Lieb-Robinson bound for the spectral flow, which
is formulated with the aid of a function $F_\Psi$ defined as follows:
\begin{equation}\label{eq:definitionFpsi}
F_{\Psi}(r) = \tilde{u}_{\mu} \left( \frac{\gamma}{8 v_a} r \right) F \left( \frac{\gamma}{8 v_a} r  \right) \, ,
\end{equation} 
where 
\begin{equation}
\tilde{u}_{\mu}(x) = \left\{ \begin{array}{cc} u_{\mu}(e^2) & \mbox{for } 0 \leq r \leq e^2, \\ u_{\mu}(x) & \mbox{otherwise}. \end{array} \right.
\end{equation}
Since $F$ is uniformly integrable over $\Gamma$ and $\tilde{u}_{\mu}(r) \leq 1$, $F_{\Psi}$ satisfies (\ref{eq:uniint}).
Moreover, $F_{\Psi}$ also satisfies (\ref{eq:convc}). In fact, it is easy to check that 
$\tilde{u}_{\mu}$ is positive, non-increasing, and logarithmically super-additive.
The Lieb-Robinson velocity in the following theorem also involves the norm 
$\Vert \Psi\Vert_{F_\Psi}$ of an interaction $\Psi$ defined later in this section (see 
\eq{eq:defPsiLambda}).

\begin{thm} \label{thm:lrbpe} Let Assumptions~\ref{ass:dphi}, \ref{ass:unigap}, and \ref{ass:unilrb} hold.
Then,
\begin{equation} \label{eq:lrbspfl}
\left\|  \left[ \alpha_s^{\Lambda}(A), B \right] \right\| \leq 2 \| A \| \| B \| \min \left[ 1, g(s) \sum_{x \in X, y \in Y} F_{\Psi}(d(x,y)) \right] ,
\end{equation}
for any $A \in \mathcal{A}_X$, $B \in \mathcal{A}_Y$, and $0 \leq s \leq 1$ and $g$ is
given by
\begin{equation}
C_{F_{\Psi}} \cdot g(t) = \left\{ \begin{array}{cc}  e^{2 \| \Psi \|_{F_\Psi} C_{F_\Psi} |t|} - 1  & \mbox{if } d(X,Y) >0, \\
e^{2 \| \Psi \|_{F_\Psi} C_{F_\Psi} |t|} & \mbox{otherwise}. \end{array} \right.
\end{equation}
The number $C_{F_\Psi}$ is as in (\ref{eq:convc}) and our estimate on $\| \Psi \|_{F_\Psi}$ 
is discussed in the next subsections.
\end{thm}

%%%%%%%%%%%%%%%%%%%%%%%

\subsection{Lieb-Robinson bounds for time-dependent interactions} \label{sec:lrbtdints}

The estimate (\ref{eq:lrbspfl}) in the statement of Thereom~\ref{thm:lrbpe} can be understood as
a Lieb-Robinson bound for the spectral flow. In this section, we demonstrate that Lieb-Robinson 
bounds hold for a large class of time-dependent interactions. As in the previous section, we assume that
our models are defined on a countable set $\Gamma$ equipped with a metric. 
Let $\Phi_t$ denote a family of interactions over $\Gamma$, and, for convenience,
we will assume that $t \in [0,1]$.  Thus, for every finite $X \subset \Gamma$ and each $t \in [0,1]$,
$\Phi_t(X)^* = \Phi_t(X) \in \mathcal{A}_X$, and we will often write $\Phi_t(X) = \Phi_X(t)$.

In this case, corresponding to each finite $\Lambda \subset \Gamma$, there is a 
time-dependent local Hamiltonian which we denote by
\begin{equation}
H_{\Lambda}(t) = \sum_{X \subset \Lambda} \Phi_X(t) \, .
\end{equation}
We will assume that, for each finite $\Lambda \subset \Gamma$, $H_{\Lambda}(t)$ is a strongly continuous
map from $[0,1]$ into $\mathcal{A}_{\Lambda}$. In this case, see e.g. Theorem X.69 \cite{reed-simon}, it is well-known that
there exists a two-parameter family of unitary propagators $U_{\Lambda}(t,s)$ with 
\begin{equation} \label{eq:uevo}
\frac{d}{dt} U_{\Lambda}(t,s) = -i H_{\Lambda}(t) U_{\Lambda}(t,s) \quad \mbox{and} \quad U_{\Lambda}(s,s) = \idty \, ,
\end{equation}
the above equation holding in the strong sense. The Heisenberg dynamics corresponding
to $H_{\Lambda}(t)$ is then defined by setting
\begin{equation}
\tau_t^{\Lambda}(A) = U_{\Lambda}(t,0)^*AU_{\Lambda}(t,0) \quad \mbox{for all } A \in \mathcal{A}_{\Lambda} \, .
\end{equation}
The following Lieb-Robinson bound holds.

\begin{thm} \label{thm:lrbtdints} Let $F$ be a non-increasing, positive function satisfying (\ref{eq:uniint}) and (\ref{eq:convc})
and suppose that the interactions $\Phi_t$ satisfy
\begin{equation} \label{eq:tdphinorm}
\| \Phi \|_F = \sup_{x,y \in \Gamma} \frac{1}{F(d(x,y))} \sum_{\stackrel{Z \subset \Gamma:}{x,y \in Z}} \sup_{0 \leq t \leq 1} \| \Phi_Z(t) \| \, < \, \infty \, .
\end{equation}
Then, for any subsets $X,Y \subset \Gamma$, $A \in \mathcal{A}_X$ and $B \in \mathcal{A}_Y$ the
estimate
\begin{equation}
\left\| \left[ \tau_t^{\Lambda}(A), B \right] \right\| \leq 2 \| A \| \| B \| \min \left[ 1, g(t) \sum_{x \in X, y \in Y} F(d(x,y)) \right] ,
\end{equation}
where the function $g$ may be taken as
\begin{equation}
C_{F } \cdot g(t) = \left\{ \begin{array}{cc} e^{2 \| \Phi \|_F C_F |t|} - 1  & \mbox{if } d(X,Y) >0, \\
e^{2 \| \Phi \|_F C_F |t|} & \mbox{otherwise}, \end{array} \right.
\end{equation}
and the number $C_F$ is as in (\ref{eq:convc}).
\end{thm}

\begin{proof}
Let $X,Y \subset \Gamma$ be finite sets. Take $\Lambda \subset \Gamma$ finite with $X \cup Y \subset \Lambda$.
Define the function $f: [0,1] \to \mathcal{A}_{\Lambda}$ by setting
\begin{equation}
f(t) = \left[ U_{\Lambda}(t,0)^* U_X(t,0) A U_X(t,0)^*U_{\Lambda}(t,0), B \right] = \left[ \tau_t^{\Lambda} \left( \tilde{\tau}_t^X(A) \right), B \right] \, ,
\end{equation}
where we have introduced the notation $\tilde{\tau}_t^X(A) = U_X(t,0) A U_X(t,0)^*$. Denoting by
\begin{equation}
S_X^{\Lambda} = \{ Z \subset \Lambda : Z \cap X \neq \emptyset, Z \cap X^c \neq \emptyset \} ,
\end{equation} 
the surface of $X$, a short calculation shows that
\begin{eqnarray}
f'(t) & = &  i \left[ \tau_t^{\Lambda} \left( [H_{\Lambda}(t) -H_X(t) , \tilde{\tau}_t^X(A) ] \right) B \right]  \nonumber \\
& = & i \sum_{\stackrel{Z \subset \Lambda:}{Z \in S_X^{\Lambda}}} \left[ \tau_t^{\Lambda}(\Phi_Z(t)), f(t)\right] +  i \sum_{\stackrel{Z \in \Lambda:}{Z \in S_X^{\Lambda}}} \left[ \tau_t^{\Lambda}( \tilde{\tau}_t^X(A)), [B, \tau_t^{\Lambda}( \Phi_Z(t))] \right]  \, . \nonumber
\end{eqnarray}
As the first term above is norm-preserving, see e.g. \cite{nach22006}, the inequality
\begin{equation}
\left\| \left[ \tau_t^{\Lambda} \left( \tilde{\tau}_t^X(A) \right), B \right]  \right\| \leq \left\| \left[ A, B \right] \right\| + 
2 \| A \| \sum_{\stackrel{Z \subset \Lambda:}{Z \in S_X^{\Lambda}}}  \int_0^{|t|} \left\| \left[ \tau_s^{\Lambda}\left( \Phi_Z(s) \right), B\right] \right\| \, ds 
\end{equation}
follows. Consider now the quantity
\begin{equation}
C_B^{\Lambda}(X,t) = \sup_{\substack{A \in \mathcal{A}_X : \\ A \neq 0}} \frac{ \| [ \tau_t^{\Lambda}(A), B] \|}{ \| A \|}
\end{equation}
It is easy to see that
\begin{equation}
C_B^{\Lambda}(X,t) \leq C_B^{\Lambda}(X,0) + 2 \sum_{\stackrel{Z \subset \Lambda:}{Z \in S_X^{\Lambda}}} \sup_{0 \leq r \leq 1} \| \Phi_Z(r) \| \int_0^{|t|} C_B^{\Lambda}(Z,s) \, ds \, .
\end{equation}
{F}rom here, the argument proceeds as in the proof of Theorem 2.1 in \cite{nach22006}. 
\end{proof}

\subsection{Some notation and a lemma} \label{sec:lemma}

In this subsection, we prove a technical estimate needed in our proof of Theorem~\ref{thm:lrbpe}.
The objective is to show that the $s$-dependent generator of the unitary flow $U_\Lambda(s)$ 
has the structure of a bonafide short-range interaction. In Theorem~\ref{thm:lppl} we showed that
each term of the perturbation, i.e.,  $\Phi_X(s)$ for a given $X$, leads to a term in the
generator that can be well approximated by {\em local} self-adjoint operator supported
in $X_R$ with almost exponentially fast decay of the error as a function of $R$. A projection
$\Pi_{X_R}:\caA_\Lambda\to\caA_{X_R}$ was used to accomplish this. In this subsection and the next
we apply the same procedure to show that the differences between successive approximations
can be summed leading to a decomposition of each term  in the generator as a telescopic sum
of finitely supported terms. To define the terms in this decomposition we need a family of
projection mappings $\left( \Pi_X \right)_{X\subset\Lambda}$, and the decomposition we obtain
will depend on the choice of this family. It will be convenient to choose a family which 
is compatible with the embeddings $\caA_{\Lambda_0}\subset\caA_{\Lambda}$, for 
$\Lambda_0\subset \Lambda$, and such that each of the $\Pi_X$ are continuous
in the norm and weak topologies on $\caA_\Lambda$. We will therefore choose a family of
normal states on $\caB(\caH_x)$, or equivalently, a family of density matrices,
$\left( \rho_x \right)_{x\in \Gamma}$ so that we can define a product state on $\caA_{X^c}$ by
setting $\rho_{X^c}=\bigotimes_{x\in \Gamma \setminus X} \rho_x$. Then, for
any finite $X\subset\Lambda$, we define
\begin{equation}
\Pi_X=\id_{\caA_X}\otimes \rho_{X^c}\vert_{\caA_\Lambda}\, .
\end{equation}
Here, $\id_{\caA_X}$ is the identity map on $\caA_X$. $\Pi_X$ can be
considered as a map $\caA_\Lambda\to \caA_\Lambda$ with $\ran \Pi_X\subset
\caA_X$.

We let the dependence of $\Pi_X$ on the $\rho_x$ be implicit. All our estimates will
be uniform in the $\rho_x$. Similarly, the interaction $\Psi_\Lambda(s)$ we define
in the next subsection depends on the choice of $\rho_x$, but the estimates on
its decay will not, and the unitary flow generated by these interactions also does 
{\em not} depend on the $\rho_x$.

Fix a finite set $\Lambda \subset \Gamma$. For any $X \subset \Gamma$ and $n \geq 0$, denote by
\begin{equation}
X_n = \left\{ z \in \Gamma : d(z,X) \leq n \right\} \, ,
\end{equation}
where $d(z,X) = \min_{x \in X} d(z,x)$. Keeping with the notation from the previous subsection, for any
$A \in \mathcal{A}_X$ we set
\begin{equation} \label{eq:Deltaat0}
\Delta_{\Lambda}^0(A,s) = \int_{- \infty}^{\infty} \Pi_{X} \left( \tau_t^{H_{\Lambda}(s)}(A) \right) W_{\gamma}(t) \, dt 
\end{equation}
and
\begin{equation} \label{eq:Deltaatn}
\Delta_{\Lambda}^n(A,s) = \int_{- \infty}^{\infty} \Pi_{X_{n}} \left( \tau_t^{H_{\Lambda}(s)}(A) \right) W_{\gamma}(t) \, dt -\int_{- \infty}^{\infty} \Pi_{X_{n-1}} \left( \tau_t^{H_{\Lambda}(s)}(A) \right) W_{\gamma}(t) \, dt 
\end{equation}
for any $n \geq 1$. Since $\Lambda$ is finite, $\Delta_{\Lambda}^n(A,s) = 0$ for
large $n$. Moreover, it is clear that $\mbox{supp} \left( \Delta_{\Lambda}^n(A,s) \right) \subset X_{n} \cap \Lambda$.
In our proof of Theorem~\ref{thm:lrbpe}, we will use that
\begin{equation}
 \int_{- \infty}^{\infty} \tau_t^{H_{\Lambda}(s)}(A) W_{\gamma}(t) \, dt = \sum_{n=0}^{\infty} \Delta_{\Lambda}^n(A,s) 
\end{equation}
where the series is actually a finite sum. In fact, the following estimate is also important.

\begin{lemma} \label{lem:Delta} Under Assumptions~\ref{ass:unigap} and \ref{ass:unilrb}, let $\Lambda \subset \Gamma$ be a
finite set. For any $X \subset \Lambda$, $A \in \mathcal{A}_X$, and integer $n \geq 0$, 
\begin{equation} \label{eq:Delbd}
\left\| \Delta_{\Lambda}^n(A,s) \right\| \leq 2 \| A \| \min\left[ \| W_{\gamma} \|_1, |X| G(n-1) \right]
\end{equation}
where
\begin{equation} \label{eq:defg}
G(n) = 4 I_{\gamma} \left( \frac{n}{2 v_a} \right) + \frac{ K_a \| F \|}{av_a} e^{-an/2} 
\end{equation}
and $I_{\gamma}$ is as in Lemma~\ref{lemma:Bigwdecay}.
\end{lemma}
\begin{proof}
It is easy to see that
\begin{equation}
\left\| \Delta_{\Lambda}^0(A,s) \right\| \leq \|A \| \, \|W_{\gamma} \|_1 \quad \mbox{and} \quad \left\| \Delta_{\Lambda}^n(A,s) \right\| \leq 2 \|A \| \, \|W_{\gamma} \|_1 \, .
\end{equation}
A better estimate in $n$ is achieved by inserting and removing an identity. In fact, we need only estimate the norm of
\begin{equation}
 \int_{- \infty}^{\infty} \left( \Pi_{X_n} - \id \right) \left( \tau_t^{H_{\Lambda}(s)}(A) \right) W_{\gamma}(t) \, dt \, .
\end{equation}
To do so, we follow the same strategy as in the proof of Theorem~\ref{thm:lppl}. By Assumption~\ref{ass:unilrb}, we know that
\begin{equation}
\left\| \left[ \tau_t^{H_{\Lambda}(s)}(A), B \right] \right\| \leq K_a \| F \| |X| \| A \| e^{a v_a |t|} e^{-an} \|B \|
\end{equation}
for all $B \in \mathcal{A}_{X_n^c}$. Hence, for any $T>0$, we have that
\begin{eqnarray}
\left\|  \int_{|t| \leq T} \left( \Pi_{X_n} - \id \right) \left( \tau_t^{H_{\Lambda}(s)}(A) \right) W_{\gamma}(t) \, dt  \right\| & \leq &  \frac{1}{2} \int_{|t| \leq T} \left\| \left( \Pi_{X_n} - \id \right) \left( \tau_t^{H_{\Lambda}(s)}(A) \right) \right\| \, dt \nonumber \\
& \leq &  K_a \| F \| |X| \| A \| \int_0^T e^{a v_a t} \, dt \, e^{-an} \, ,
\end{eqnarray}
using Lemma~\ref{lma:Approx}, whereas
\begin{equation}
\left\|  \int_{|t| > T} \left( \Pi_{X_n} - \id \right) \left( \tau_t^{H_{\Lambda}(s)}(A) \right) W_{\gamma}(t) \, dt  \right\|  \leq 4 \| A \| I_{\gamma}(T) \, .
\end{equation}
The choice of $T = n/2v_a$ yields an estimate of the form
\begin{equation}
\left\|  \int_{- \infty}^{\infty} \left( \Pi_{X_n} - \id \right) \left( \tau_t^{H_{\Lambda}(s)}(A) \right) W_{\gamma}(t) \, dt \right\| \leq 4 \|A\| I_{\gamma} \left( \frac{n}{2v_a} \right)
+ \frac{K_a \| F \|}{a v_a} |X| \| A \| e^{-an/2} \, .
\end{equation}
The bound (\ref{eq:Delbd}) readily follows.
\end{proof}

As indicated by the proof above, a stronger inequality is true. We have actually shown that
for every $n \geq 1$,
\begin{equation}
\left\| \Delta_{\Lambda}^n(A,s) \right\| \leq 2 \| A \| \min\left[ \| W_{\gamma} \|_1, G_A(n-1) +G_A(n) \right]
\end{equation}
where
\begin{equation}
G_A(n) = 2 I_{\gamma} \left( \frac{n}{2 v_a} \right) + \frac{ K_a \| F \|}{2av_a} |X| e^{-an/2} \, .
\end{equation}
For the arguments we use below, it is convenient to extract a decaying quantity that is independent of
the given observable $A$ and use the monotonicity of $G$. This explains the form of the bound
(\ref{eq:Delbd}) appearing in Lemma~\ref{lem:Delta}.

\subsection{The proof of Theorem~\ref{thm:lrbpe}} \label{sec:proof}

In this subsection, we prove Theorem~\ref{thm:lrbpe}. 
The basic idea is that Theorem~\ref{thm:lrbpe} follows
from a Lieb-Robinson bound for time-dependent interactions, see
e.g. Theorem~\ref{thm:lrbtdints} in Section~\ref{sec:lrbtdints}. 
To see that such a result is applicable, we demonstrate that the
generator of the spectral flow can be written as a
sum of local interaction terms which satisfy an appropriate
decay assumption. This is the content of Theorem~\ref{thm:intdec} below.

Under Assumptions~\ref{ass:dphi} and \ref{ass:unigap}, we have defined (for each finite $\Lambda \subset \Gamma$)
a spectral flow by setting
\begin{equation}
\alpha_s^{\Lambda}(A) = U_{\Lambda}(s)^* A U_{\Lambda}(s) \quad \mbox{for all } A \in \mathcal{A}_{\Lambda}.
\end{equation}
In fact, the unitary $U_{\Lambda}(s)$ is the one constructed in Proposition~\ref{prop:main}, and
as a consequence of Corollary~\ref{cor:DW}, we know that $U_{\Lambda}(s)$ is generated by
\begin{eqnarray} \label{eq:dsum1}
D_{\Lambda}(s) & = &  \int_{-\infty}^{\infty} \tau_t^{H_{\Lambda}(s)} \left( H_{\Lambda}'(s) \right) W_{\gamma}(t) \, dt  \nonumber \\
& = & \sum_{Z \subset \Lambda} \int_{-\infty}^{\infty} \tau_t^{H_{\Lambda}(s)} \left( \Phi_Z'(s) \right) W_{\gamma}(t) \, dt  \, .
\end{eqnarray}
Here $\gamma$ is as in Assumption~\ref{ass:unigap}, and 
$W_{\gamma}$ appears in Lemma~\ref{lemma:Bigwdecay} . The previous
subsection demonstrated that each term
\begin{equation} \label{eq:delsum}
\int_{-\infty}^{\infty} \tau_t^{H_{\Lambda}(s)} \left( \Phi_Z'(s) \right) W_{\gamma}(t) \, dt = \sum_{n=0}^{\infty} \Delta_{\Lambda}^n( \Phi_Z'(s),s) 
\end{equation}
where the series is actually a finite sum. Combining (\ref{eq:dsum1}) and (\ref{eq:delsum}) above, we write
\begin{equation}
D_{\Lambda}(s) = \sum_{Z \subset \Lambda} \sum_{n=0}^{\infty} \Delta_{\Lambda}^n( \Phi_Z'(s),s) = \sum_{Z \subset \Lambda} \Psi_{\Lambda}(Z,s) \, ,
\end{equation}
where
\begin{equation}\label{eq:defPsiLambda}
\Psi_{\Lambda}(Z,s) = \sum_{n \geq 0} \sum_{\substack{Y \subset \Lambda: \\ Y_n =Z}} \Delta_{\Lambda}^n( \Phi_Y'(s),s) \, .
\end{equation}
It is important here to note that $\mbox{supp}( \Psi_{\Lambda}(Z,s)) \subset Z$, i.e., the $s$-dependent, interaction terms $\Psi_{\Lambda}(Z,s)$
are strictly local. The following estimate holds.

\begin{thm} \label{thm:intdec} 
Let Assumptions~\ref{ass:dphi}, \ref{ass:unigap}, and \ref{ass:unilrb} hold. Then,
there exists a function $F_{\Psi}$ satisfying (\ref{eq:uniint}) and (\ref{eq:convc}) such that
\begin{equation} \label{eq:intbdthm}
\| \Psi_{\Lambda} \|_{F_{\Psi}} = \sup_{x, y \in \Lambda} \frac{1}{F_{\Psi}(d(x,y))} \sum_{\stackrel{Z \subset \Lambda:}{x, y \in Z}} \sup_{0 \leq s \leq 1} \| \Psi_{\Lambda}(Z,s) \| < \infty \, .
\end{equation}
Here we note that the function $F_{\Psi}$ is independent of $\Lambda$.
\end{thm}

It is now clear that Theorem~\ref{thm:lrbpe} follows from Theorem~\ref{thm:intdec} via an application of Theorem~\ref{thm:lrbtdints}.

\begin{proof}
In the argument below, it is convenient to set $a>0$ to be the minimum of the $a$'s whose
existences are guaranteed by Assumptions~\ref{ass:dphi} and \ref{ass:unilrb}.

We begin by re-writing the quantity of interest. Clearly,
\begin{equation}
\sup_{0 \leq s \leq 1} \| \Psi_{\Lambda}(Z,s) \|  \leq \sum_{\stackrel{Y,n \geq 0:}{Y_n =Z}} \sup_{0 \leq s \leq 1} \| \Delta_{\Lambda}^n( \Phi_Y'(s),s) \| \, ,
\end{equation}
and so 
\begin{eqnarray} \label{eq:psibd1}
\sum_{\stackrel{Z \subset \Lambda:}{x,y \in Z}} \sup_{0 \leq s \leq 1} \| \Psi_{\Lambda}(Z,s) \|  & \leq & \sum_{\stackrel{Z \subset \Lambda:}{x,y \in Z}} \sum_{\stackrel{Y,n \geq 0:}{Y_n =Z}} \sup_{0 \leq s \leq 1} \| \Delta_{\Lambda}^n( \Phi_Y'(s),s) \| \nonumber \\ 
& = & \sum_{Y \subset \Lambda} \sum_{n \geq 0} {\rm Ind} \left[ x,y \in Y_n \right]  \sup_{0 \leq s \leq 1} \| \Delta_{\Lambda}^n( \Phi_Y'(s),s) \| \nonumber \\
&=&  \sum_{\stackrel{Y \subset \Lambda:}{x,y \in Y}} \sum_{n \geq 0}  \sup_{0 \leq s \leq 1} \| \Delta_{\Lambda}^n( \Phi_Y'(s),s) \|  \\ 
& \mbox{ } & + \sum_{m=1}^{\infty} \sum_{\stackrel{Y \subset \Lambda:}{x,y \in Y_m}} {\rm Ind}\left[ \{x,y\} \cap Y_{m-1}^c \neq \emptyset \right] \sum_{n \geq m} \sup_{0 \leq s \leq 1} \| \Delta_{\Lambda}^n( \Phi_Y'(s),s) \| \, \nonumber .
\end{eqnarray}
The first equality above follows from the observation that
\begin{eqnarray}
\sum_{\stackrel{Z \subset \Lambda:}{x,y \in Z}} \sum_{\stackrel{Y,n \geq 0:}{Y_n =Z}}  
& = & \sum_{Z \subset \Lambda} \sum_{Y \subset \Lambda} \sum_{n \geq 0} {\rm Ind} \left[ Y_n = Z \right] {\rm Ind} \left[ x,y \in Z \right]  \nonumber \\ 
& = &  \sum_{Y \subset \Lambda} \sum_{n \geq 0} \left[ \sum_{Z \subset \Lambda} {\rm Ind} \left[ Y_n = Z \right] \right] {\rm Ind} \left[ x,y \in Y_n  \right]  \nonumber \\
& = &  \sum_{Y \subset \Lambda} \sum_{n \geq 0} {\rm Ind} \left[ x,y \in Y_n  \right] \, , 
\end{eqnarray}
while the second is a consequence of the fact that for any pair $x,y$ 
\begin{equation}
\sum_{Y \subset \Lambda} = \sum_{\stackrel{Y \subset \Lambda:}{x,y \in Y}} + \sum_{m \geq 1} \sum_{\stackrel{Y \subset \Lambda:}{x,y \in Y_m}} {\rm Ind}\left[ \{x,y\} \cap Y_{m-1}^c \neq \emptyset \right] \, .
\end{equation}

The first sum on the right-hand-side of (\ref{eq:psibd1}) is easy to bound. In fact, using Lemma~\ref{lem:Delta}, it is
clear that
\begin{equation} \label{eq:Delnbd}
\sup_{0 \leq s \leq 1} \| \Delta_{\Lambda}^n( \Phi_Y'(s),s) \| \leq 2 |Y| \sup_{0 \leq s \leq 1} \| \Phi_Y'(s) \| G(n-1) \, ,
\end{equation}
where $G$ is as in (\ref{eq:defg}) with $G(-1)$ set to be $\| W_{\gamma} \|_1$.
Thus,
\begin{eqnarray}
 \sum_{\stackrel{Y \subset \Lambda:}{x,y \in Y}} \sum_{n \geq 0}  \sup_{0 \leq s \leq 1} \| \Delta_{\Lambda}^n( \Phi_Y'(s),s) \| & \leq & 
 2 \sum_{n \geq 0}  G(n-1) \sum_{\stackrel{Y \subset \Lambda:}{x,y \in Y}} |Y| \sup_{0 \leq s \leq 1} \| \Phi_Y'(s) \| \nonumber \\
 & \leq & 2 \| \partial \Phi \|_a F_a(d(x,y)) \sum_{n \geq 0}  G(n-1) \, .
\end{eqnarray}
{F}rom the estimates in Lemma~\ref{lemma:Bigwdecay}, it is clear that $G$ is summable.

For the remaining terms in (\ref{eq:psibd1}), we use the following over-counting estimate:
\begin{equation} \label{eq:overcount}
\sum_{\stackrel{Y \subset \Lambda:}{x,y \in Y_m}} {\rm Ind}\left[ \{x,y\} \cap Y_{m-1}^c \neq \emptyset \right] \leq  \sum_{y_1 \in B_m(x)} \sum_{y_2 \in B_m(y)} \sum_{\stackrel{Y \subset \Lambda:}{y_1,y_2 \in Y}} 
\end{equation}
Combining (\ref{eq:Delbd}) with (\ref{eq:overcount}), we find that 
\begin{eqnarray} \label{eq:summbd}
 \sum_{m \geq 1} \sum_{\stackrel{Y \subset \Lambda:}{x,y \in Y_m}} {\rm Ind}\left[ \{x,y\} \cap Y_{m-1}^c \neq \emptyset \right] \sum_{n \geq m} \sup_{0 \leq s \leq 1} \| \Delta_{\Lambda}^n( \Phi_Y'(s),s) \| 
  \nonumber \\
 \leq 2 \sum_{m \geq 1}  \sum_{y_1 \in B_m(x)} \sum_{y_2 \in B_m(y)} \sum_{\stackrel{Y \subset \Lambda:}{y_1,y_2 \in Y}} |Y| \sup_{0 \leq s \leq 1} \| \Phi_Y'(s) \| \sum_{n \geq m} G(n-1) \nonumber \\
 \leq 2 \| \partial \Phi \|_a  \sum_{m \geq 1}  \hat{G}(m) \sum_{y_1 \in B_m(x)} \sum_{y_2 \in B_m(y)} F_a(d(y_1,y_2)) \, ,
\end{eqnarray}
where we have set 
\begin{equation}
\hat{G}(m) = \sum_{n \geq m} G(n-1) \, .
\end{equation}

We now perform a rough optimization over $m \geq 1$. Take $0 < \epsilon < 1$ and declare $m_0 = m_0(\epsilon) \geq 0$ to be the largest integer
less than $(1 - \epsilon)d(x,y)/2$. We claim that, for $m \leq m_0$ and $y_1$ and $y_2$ as in (\ref{eq:summbd}) above, $\epsilon d(x,y) \leq d(y_1,y_2)$. This follows from
\begin{equation}
d(x,y) \leq d(x,y_1) + d(y_1,y_2) + d(y_2, y) \leq d(y_1,y_2) + 2 m \leq d(y_1,y_2) + 2 m_0 \, ,
\end{equation}
and the choice of $m_0$. In this case we have
\begin{eqnarray}
\sum_{m = 1}^{m_0+1}  \hat{G}(m) \sum_{y_1 \in B_m(x)} \sum_{y_2 \in B_m(y)} F_a(d(y_1,y_2)) & \leq & \hat{G}(1) F_a( \epsilon d(x,y)) \sum_{m=1}^{m_0+1} |B_m(x)| |B_m(y)| \nonumber \\
& \leq & \kappa^2 \hat{G}(1) F_a( \epsilon d(x,y) ) \sum_{m=1}^{m_0+1} m^{2 \nu} \, ,
\end{eqnarray}
where we have used (\ref{eq:ballgrow}).

The remaining terms we bound as follows.
\begin{eqnarray}
\sum_{m > m_0+1}  \hat{G}(m) \sum_{y_1 \in B_m(x)} \sum_{y_2 \in B_m(y)} F_a(d(y_1,y_2)) & \leq & \| F_a \| \sum_{m > m_0+1} |B_m(x)| \hat{G}(m) \nonumber \\
& \leq & \kappa \| F_a \| \sum_{m >m_0+1} m^{ \nu} \hat{G}(m) \, .
\end{eqnarray}
Now, from the definition of $\hat{G}$,
\begin{equation}
\sum_{m >m_0+1} m^{ \nu} \hat{G}(m) = \sum_{m=m_0+2}^{\infty} m^{\nu} \sum_{n=m-1}^{\infty} \left( 4I_{\gamma} \left( \frac{n}{2 v_a} \right) + \frac{K_a \| F \|}{a v_a} e^{- an/2} \right) \, ,
\end{equation}
and the sum
\begin{equation}
\sum_{m=m_0+2}^{\infty} m^{\nu} \sum_{n=m-1}^{\infty} e^{-an/2} = e^{a/2} \sum_{y \geq 0} e^{-ay/2} \cdot \sum_{m=m_0+2} m^{\nu} e^{-am/2}
\end{equation}
decays exponentially in $m_0$. Using the results in Lemma~\ref{lem:ua} and \ref{lemma:Bigwdecay}, we find
that
\begin{eqnarray}
\sum_{m=m_0+2}^{\infty} m^{\nu} \sum_{n=m-1}^{\infty} I_{\gamma} \left( \frac{n}{2 v_a} \right)  & \leq & 
\frac{C}{\gamma} \sum_{m=m_0+2}^{\infty} m^{\nu} \sum_{n=m-1}^{\infty} \left( \frac{\gamma n}{ 2 v_a} \right)^{10} u_{2/7} \left( \frac{ \gamma n}{ 2 v_a} \right) \nonumber \\
& \leq & \frac{2 v_a C}{\gamma^2} \sum_{m=m_0+2}^{\infty} m^{\nu} \int_{\frac{\gamma(m-1)}{2 v_a}}^{\infty} y^{10} u_{2/7}(y) \, dy \nonumber \\
& \leq & \frac{161 v_a C}{\gamma^2} \sum_{m=m_0+2}^{\infty} m^{\nu} \left( \frac{\gamma(m-1)}{2 v_a} \right)^{22} u_{2/7} \left( \frac{\gamma(m-1)}{2 v_a}  \right) \nonumber \\
& \leq & \frac{322 v_a^2 C}{ \gamma^3} \int_{\frac{ \gamma (m_0+1)}{2 v_a}}^{\infty} \left( \frac{2 v_a y}{\gamma} + 1 \right)^{\nu} y^{22} u_{2/7}(y) \, dy \nonumber \\
& \leq &  \frac{2254 \cdot 2^{2 \nu}}{\gamma} \left( \frac{v_a}{\gamma} \right)^{\nu+2} (47+2 \nu)  \left( \frac{ \gamma (m_0+1)}{2 v_a} \right)^{46+2\nu} u_{2/7}\left( \frac{\gamma (m_0+1)}{2 v_a} \right) \, .
\end{eqnarray}
This proves that 
\begin{eqnarray}
\sum_{\stackrel{Z \subset \Lambda:}{x,y \in Z}} \sup_{0 \leq s \leq 1} \| \Psi_{\Lambda}(Z,s) \| & \leq & C_1 F_a( \epsilon d(x,y)) (m_0+1)^{2 \nu +1} + C_2 \sum_{m=m_0+2} m^{\nu} e^{- am/2} \nonumber \\
& \mbox{ } & \quad + C_3 \left( \frac{\gamma (m_0+1)}{2 v_a} \right)^p u_{2/7}\left( \frac{ \gamma (m_0+1)}{2 v_a} \right) \, 
\end{eqnarray}
for some number $p$ depending only on $\nu$. Since $2m_0 \leq (1 - \epsilon)d(x,y)$, it is clear that the final term above decays the slowest in $d(x,y)$. Thus we have shown that
\begin{equation}
\sum_{\stackrel{Z \subset \Lambda:}{x,y \in Z}} \sup_{0 \leq s \leq 1} \| \Psi_{\Lambda}(Z,s) \| \leq C  \left( \frac{\gamma}{2 v_a} \left( \frac{1- \epsilon}{2}d(x,y) + 1 \right) \right)^p 
u_{2/7}\left( \frac{ \gamma (1- \epsilon)}{4 v_a} d(x,y) \right) \, ,
\end{equation}
for each $0 < \epsilon <1$. For concreteness, take $\epsilon = 1/2$.
With $\delta>0$ as in (\ref{eq:Fgrowth}) and any $0< \delta' < 2/7- \delta$, we will set $\mu = 2/7 - \delta - \delta'>0$ and see that
\begin{eqnarray}
\sum_{\stackrel{Z \subset \Lambda:}{x,y \in Z}} \sup_{0 \leq s \leq 1} \| \Psi_{\Lambda}(Z,s) \| & \leq & C'  \left( \frac{5 \gamma}{8 v_a} d(x,y) \right)^p u_{2/7}\left( \frac{ \gamma}{8 v_a} d(x,y) \right) \nonumber \\
& \leq & C'' \left( \frac{5 \gamma}{8 v_a} d(x,y) \right)^p u_{2/7- \delta}\left( \frac{ \gamma}{8 v_a} d(x,y)  \right) F \left( \frac{ \gamma}{8 v_a} d(x,y)  \right)  \nonumber \\
& \leq & C''' u_{\mu}\left( \frac{ \gamma}{8 v_a} d(x,y)  \right) F \left( \frac{ \gamma}{8 v_a} d(x,y)  \right) \, .
\end{eqnarray}
With the definition of $F_\Psi$ given in \eq{eq:definitionFpsi}, this completes the proof 
of \eq{eq:intbdthm}.
\end{proof}

%%%%%%%%%%%%%%%%%%%%%%%%%%%%%%%%%%%%%%%%%%%%%%%%%%%%%%%%%%%%%%%%%%%%%%%%%%%%%%%%%%%%%%%%%%%%%%%%%%%%%%%%

\section{Existence of the thermodynamic limit and gapped quantum phases}
\label{sec:quantum_phases}

The Lieb-Robinson bound for the flow $\alpha_s^\Lambda$ given in Theorem \ref{thm:lrbpe} 
of the previous section, can be used to obtain the thermodynamic limit of this flow defined as
a strongly continuous cocycle of automorphisms of the $C^*$-algebra of quasi-local observables.
The standard setting is the same as in the previous section, but we now assume that
the Hilbert spaces $\mathcal{H}_x$ associated to each $x \in \Gamma$, are all finite-dimensional.
The $C^*$-algebra of quasi-local observables $\mathcal{A}_{\Gamma}$ is then obtained as the completion 
with respect to the operator norm of $\mathcal{A}_{\rm loc}$:
\begin{equation}
\mathcal{A}_{\Gamma} =\overline{\mathcal{A}_{\rm loc}} = \overline{\bigcup_{\Lambda \subset \Gamma} 
\mathcal{A}_{\Lambda}}.
\end{equation}
If $\mathcal{H}_x$ is allowed to be infinite-dimensional it is typically necessary to 
work in the GNS representation of a reference state in order to have a well-defined
thermodynamic limit. Such an approach was used in \cite{Nach10} to define the 
dynamics of an infinite lattice of anharmonic oscillators. In order to avoid the need for
additional technical assumptions, for the remainder of this section
we restrict ourselves to quantum spin systems, i.e., the case of finite-dimensional $\mathcal{H}_x$.
It is not necessary, however, that $\dim \mathcal{H}_x$ is independent of $x$ or even 
uniformly bounded.

This section has two subsections. In the first, we prove that the finite volume
spectral flows, defined as in (\ref{automorphism}), have a 
well-defined thermodynamic limit. With these results in hand we can then, in the second 
subsection, complete the proof that gapped ground states connected 
by a curve of quasi-local interactions satisfying a suitable norm condition are equivalent 
under a quasi-local automorphism, in finite volume as well as in the thermodynamic
limit. But first we describe in detail the class of systems to which our main result applies.

The systems under consideration here have finite dimensional local Hilbert spaces.
In this case, we can make a convenient choice of the projection map introduced in 
Section~\ref{sec:lemma} and needed for the application of Lemma~\ref{lma:Approx}, 
namely the natural extension of the partial trace. For any finite subset 
$\Lambda\subset\Gamma$, we define the conditional expectation 
$\Pi_\Lambda: \caA_{\Gamma} \rightarrow\caA_\Lambda$ as
\begin{align*}
\Pi_\Lambda = \mathrm{id}_{\caA_\Lambda}\otimes\tau_{\caA_{\Lambda^c}}\,,
\end{align*}
where for $\Lambda'\subset\Gamma$, 
\begin{equation*}
\tau_{\caA_{\Lambda'}} = \bigotimes_{x\in \Lambda'}\tau_{\caA_{x}}\,,\qquad \tau_{\caA_{x}} = \frac{1}{\mathrm{dim}\caH_x}\Tr_{\caH_x}
\end{equation*}
is the normalized trace over $\caA_{\Lambda'}$.
In particular, for any $Z\subset\Lambda\subset\Gamma$, the subprojections
\begin{equation*}
\Pi_{\Lambda,Z} = \left.\Pi_{Z}\right|_{\caA_{\Lambda}}
\end{equation*}
form a consistent family, namely for any $A\in\caA_{X}$, with $Z, X\subset \Lambda_m\subset\Lambda_n\subset\Gamma$, they satisfy
\begin{equation} \label{ConsistentProj}
\Pi_{\Lambda_n,Z}(A) = \Pi_{\Lambda_m,Z}(A)
\end{equation}
and the first index may be dropped.

Let $\Gamma$ be a countable set equipped with a metric and a function $F$
satisfying (\ref{eq:uniint}) and (\ref{eq:convc}). For $s \in [0,1]$, let $\Phi(s)$ be a
family of interactions, differentiable in $s$, for which there exists a number $a>0$ so that
\begin{equation} \label{eq:intbd}
\| \Phi \|_a + \| \partial \Phi \|_a  < \infty \, .
\end{equation} 
where the norm is defined in the paragraph containing \eq{Fnorm}.

Our proof of the existence of the thermodynamic limit requires some 
assumptions on the sequence of finite volumes $\left( \Lambda_n \right)_n$
on which the spectral flows are defined. 
Let $\left( \Lambda_n \right)_n$ be an increasing sequence of finite sets
which exhaust $\Gamma$ as $n\to \infty$. For convenience, we will regard the 
parameter $n$ as continuous with the understanding that, for any $n \geq 0$,
$\Lambda_n=\Lambda_{[n]}$, where $[n]$ denotes the integer
part of $n$. We will assume that there exist positive numbers $b_1$, $b_2$, and
$p$ such that
\begin{equation}\label{conditionsLambda_m}
d(\Lambda_m, \Lambda_n^c) \geq b_1(n-m),\quad \mbox{and }
| \Lambda_n| \leq b_2 n^p\, .
\end{equation}

We assume that there are finite intervals $I(s)$, smoothly depending on $s\in[0,1]$ such that, for all $n$, the finite-volume Hamiltonians $H_{\Lambda_n}(s) = \sum_{Z\subset\Lambda_n}\Phi(Z,s)$ have one or more eigenvalues in $I(s)$, and no eigenvalues outside $I(s)$ within a distance $\gamma>0$ of it.

Let us summarize the results of the previous sections, given these assumptions. If $P_{\Lambda_n}(s)$ denotes the spectral projections of $H_{\Lambda_n}(s)$ on $I(s)$, then there is a cocycle $\alpha_s^{\Lambda_n}$, the dual of which maps $P_{\Lambda_n}(0)$ to $P_{\Lambda_n}(s)$ for all $s\in[0,1]$. Its generator has a local structure given by $D_{\Lambda_n}(s) = \sum_{Z\subset\Lambda_n}\Psi_{\Lambda_n}(Z,s)$ where the interactions $\Psi_{\Lambda_n}(s)$ decay almost exponentially in the following sense,
\begin{equation}
\| \Psi_{\Lambda_n} \|_{F_{\Psi}} = \sup_{x,y\in{\Lambda_n}}\frac{1}{F_\Psi(d(x,y))}\sum_{\substack{Z\subset {\Lambda_n} \\ x,y \in Z}}\sup_{0\leq s\leq 1}\Vert\Psi_{\Lambda_n}(Z,s)\Vert< \infty\,,
\label{triple_norm_psi}\end{equation}
uniformly in $n$, where $F_\Psi$ satisfies again the uniform integrability and convolution property for a constant $C_\Psi$. 
Our estimates in Section~\ref{Sec:LRforD} demonstrate that a possible choice of $F_{\Psi}$ is given by (\ref{eq:definitionFpsi})
which decays sub-exponentially.

%%%%%%%%%%%%%%%%%%%%%%%%%%%%%
\subsection{Thermodynamic limit for the spectral flow}
\label{sec:thermodynamic_limit}

In order to prove the existence of the thermodynamic limit of the 
spectral flow $\alpha_s^{\Lambda}$, it is convenient to recall
an estimate from the proof of the existence of the thermodynamic limit of
Heisenberg evolutions $\tau_t^{H_{\Lambda}(s)}$, as proven e.g. in \cite{nach22006}.
In fact,  assuming that $\| \Phi \|_a < \infty$, the following bound is valid.

Take finite sets $X \subset \Lambda_m \subset \Lambda_n$. Note that
for any $A \in \mathcal{A}_X$, each $s \in [0,1]$, and any $t \in \mathbb{R}$,
\begin{eqnarray}
\left\| \tau_t^{H_{\Lambda_n}(s)}(A) - \tau_t^{H_{\Lambda_m}(s)}(A) \right\| & \leq & \sum_{\stackrel{Z \subset \Lambda_n:}{Z \cap \Lambda_n \setminus \Lambda_m \neq \emptyset}}
\int_0^{|t|} \left\| \left[ \Phi_Z(s), \tau_{|t|-r}^{H_{\Lambda_m}(s)}(A) \right] \right\| \, dr \, \nonumber \\
& \leq & \frac{K_a \|A \|}{a v_a} (e^{av_a|t|}-1) \sum_{\stackrel{Z \subset \Lambda_n:}{Z \cap \Lambda_n \setminus \Lambda_m \neq \emptyset}} \| \Phi_Z(s) \| \sum_{z \in Z, x \in X}F_a(d(x,z)) \nonumber \\
& \leq &   \frac{K_a \|A \|}{a v_a} C_{F_a} \| \Phi \|_a (e^{av_a|t|}-1) \sum_{\stackrel{x \in X}{y \in \Lambda_n \setminus \Lambda_m}} F_a(d(x,y)) \, . 
\end{eqnarray}
Since $F_a$ is uniformly integrable, this proves that the sequence $\left( \tau_t^{H_{\Lambda_n}(s)}(A) \right)_n$ is Cauchy. We will
denote the limit by $\tau_t^{\Gamma,s}(A)$, and observe that it satisfies 
\begin{equation} \label{eq:thermest}
\left\| \tau_t^{\Gamma , s}(A) - \tau_t^{H_{\Lambda_m}(s)}(A) \right\|  \leq  \frac{K_a \|A \|}{a v_a} C_{F_a} \| \Phi \|_a (e^{av_a|t|}-1) \sum_{\stackrel{x \in X}{y \in \Gamma \setminus \Lambda_m}} F_a(d(x,y)) \, ,
\end{equation}
uniformly for $s \in [0,1]$.

The following analogue of Lemma~\ref{lem:Delta} will be useful. Recall the definitions of $\Delta_{\Lambda}^n(A,s)$ from (\ref{eq:Deltaat0}) and (\ref{eq:Deltaatn}).
Define similarly $\Delta_{\Gamma}(A,s)$ with $\tau_t^{\Gamma,s}(A)$ replacing $\tau_t^{H_{\Lambda}(s)}(A)$ as appropriate.

\begin{lemma} \label{lem:Delbd2}
Let $\Lambda \subset \Gamma$ be a
finite set. For any $X \subset \Lambda$ and $A \in \mathcal{A}_X$, 
\begin{equation} \label{eq:Delbd2}
\left\| \Delta_{\Lambda}^n(A,s) - \Delta_{\Gamma}^n(A,s) \right\| \leq 4 \| A \| \min\left[ \| W_{\gamma} \|_1, |X| \sqrt{G(n-1) K(d(X, \Lambda^c))} \right]
\end{equation}
where $G$ is as in (\ref{eq:defg}) of Lemma~\ref{lem:Delta} and
\begin{equation}
K(x) = 4 I_{\gamma} \left( \frac{x }{2 v_a} \right) + \frac{K_a C_{F_a} \| \Phi \|_a \| F \|}{a^2 v_a^2} e^{-a x/2}.
\end{equation}
\end{lemma}

\begin{proof}
A uniform estimate, as shown in Lemma~\ref{lem:Delta}, clearly holds for $n=0$. We need only consider $n \geq 1$.
Using the consistency of the mappings $\Pi_{X_n}$, the difference $\Delta_{\Lambda}^n(A,s) - \Delta_{\Gamma}^n(A,s)$
can be written as a difference of two terms. As such, we need only bound the norm of
\begin{equation}
\int_{- \infty}^{\infty} \left( \Pi_{X_n} - \id \right) \left( \tau_t^{\Gamma,s}(A) - \tau_t^{H_{\Lambda}(s)}(A) \right) W_{\gamma}(t) \, dt \, .
\end{equation}
By Assumption~\ref{ass:unilrb}, $\tau_t^{H_{\Lambda}(s)}$ satisfies a Lieb-Robinson bound uniform in $\Lambda$ and $s$.
In this case, the limit $\tau_t^{\Gamma,s}$ does as well. Arguing then as in Lemma~\ref{lem:Delta}, it is clear that
\begin{eqnarray} \label{eq:est1}
\left\| \int_{- \infty}^{\infty} \left( \Pi_{X_n} - \id \right) \left( \tau_t^{\Gamma,s}(A) - \tau_t^{H_{\Lambda}(s)}(A) \right) W_{\gamma}(t) \, dt \, \right\| \hspace{4cm} \nonumber \\ \hspace{2cm} \leq 2 | X| \| A \| \left( 4 I_{\gamma} \left( \frac{n}{2 v_a} \right) + \frac{K_a \| F \|}{a v_a} e^{-an/2} \right) \, .
\end{eqnarray}

Since the projections $\Pi_{X_n}$ are norm one maps, we may also argue using the thermodynamic estimate (\ref{eq:thermest}).
In fact, 
\begin{equation*}
\left\| \int_{- \infty}^{\infty} \left( \Pi_{X_n} - \id \right) \left( \tau_t^{\Gamma,s}(A) - \tau_t^{H_{\Lambda}(s)}(A) \right) W_{\gamma}(t) \, dt \, \right\| \leq 
2  \int_{- \infty}^{\infty} \left\| \tau_t^{\Gamma,s}(A) - \tau_t^{H_{\Lambda}(s)}(A) \right\| \left| W_{\gamma}(t) \right| \, dt \, .
\end{equation*}
Now for $|t| \leq T$, we have that
\begin{eqnarray}
2 \int_{|t| \leq T} \left\| \tau_t^{\Gamma,s}(A) - \tau_t^{H_{\Lambda}(s)}(A) \right\| \left| W_{\gamma}(t) \right| \, dt & \leq &  \frac{K_a C_{F_a} \| \Phi \|_a}{a v_a}  \|A \| \sum_{\stackrel{x \in X}{y \in \Gamma \setminus \Lambda}} F_a(d(x,y)) \int_{|t| \leq T} e^{a v_a |t|} \, dt \nonumber \\ 
& \leq &  \frac{2 K_a C_{F_a} \| \Phi \|_a}{a^2 v_a^2}  \|A \| |X| \| F \| e^{-a d(X, \Lambda^c)}  e^{av_aT} \, ,
\end{eqnarray}
whereas for $|t| >T$, the bound
\begin{equation}
2 \int_{|t| > T} \left\| \tau_t^{\Gamma,s}(A) - \tau_t^{H_{\Lambda}(s)}(A) \right\| \left| W_{\gamma}(t) \right| \, dt  \leq 8 \| A \| I_{\gamma}(T) \,,
\end{equation}
is clearly true. In this case, the choice $T = d(X, \Lambda^c)/(2v_a)$ yields the estimate
\begin{eqnarray} \label{eq:est2}
\left\| \int_{- \infty}^{\infty} \left( \Pi_{X_n}  - \id \right) \left( \tau_t^{\Gamma,s}(A) - \tau_t^{H_{\Lambda}(s)}(A) \right) W_{\gamma}(t) \, dt \, \right\|  \hspace{4cm} \nonumber \\ \hspace{1cm} \leq 2 | X| \| A \|  \left( 4 I_{\gamma} \left( \frac{d(X, \Lambda^c) }{2 v_a} \right) + \frac{K_a C_{F_a} \| \Phi \|_a \| F \|}{a^2 v_a^2} e^{-a d(X,\Lambda^c)/2} \right) \, .
\end{eqnarray}
Combining the results from (\ref{eq:est1}) and (\ref{eq:est2}), as well as the bound corresponding to $\Pi_{X_{n-1}}$, the estimate (\ref{eq:Delbd}) follows.
\end{proof}

We can now state and prove the existence of the thermodynamic limit for the spectral flow $\alpha_s^{\Lambda_n}$.
Recall that for any finite sets $Z \subset \Lambda \subset \Gamma$, we have defined 
\begin{equation}
\Psi_{\Lambda}(Z,s) = \sum_{\stackrel{Y,n \geq 0:}{Y_n=Z}} \Delta_{\Lambda}^n( \Phi_Y'(s),s) \, .
\end{equation}
By analogy, set
\begin{equation} \label{eq:PsiG}
\Psi_{\Gamma}(Z,s) = \sum_{\stackrel{Y,n \geq 0:}{Y_n=Z}} \Delta_{\Gamma}^n( \Phi_Y'(s),s) \, .
\end{equation}
We will show later in this subsection that the $s$-dependent interaction $\Psi_\Gamma(s)$ is the limit as $\Lambda\to\Gamma$
of $\Psi_\Lambda(s)$.  First, we show the existence of the limiting spectral flow $\alpha^\Gamma_s$
in Theorem \ref{thm:ThermoLimit}. Then, we argue that it is also the limit of the automorphisms
generated by finite volume restrictions of the limiting interaction $\Psi_\Gamma(s)$.

\begin{thm} \label{thm:ThermoLimit}
Let $\left(\alpha_s^{\Lambda_n}\right)_n$ denote the sequence of flows associated with the sets $\Lambda_n\subset\Gamma$. Then there exists a flow $\alpha_s^\Gamma$ defined on the quasi-local algebra $\caA_{\Gamma}$ such that for all $A\in\caA_{\rm loc}$,
\begin{equation*}
\lim_{n\to\infty}\|\alpha_s^{\Lambda_n}(A)-\alpha_s^{\Gamma}(A)\| = 0\,,
\end{equation*}
uniformly for all $s\in[0,1]$.
\end{thm}

\begin{proof}
We begin by noting that the strong limit of an automorphism is automatically an
automorphism and that convergence of a sequence of automorphisms $\sigma_n\to\sigma$,
is equivalent to the convergence of the inverses to the inverse automorphism, i.e.,
$\sigma_n^{-1}\to \sigma^{-1}$. Using these observations and by standard completeness arguments it is therefore sufficient to establish that for all $A\in \mathcal{A}_{\rm loc}$,
the sequence $(\alpha^{\Lambda_n}_s)^{-1}(A)$ is Cauchy. Without loss of generality, we assume that $A\in\caA_{\Lambda_0}$ and we use the notation
$\tilde\alpha_s^{\Lambda_n}=(\alpha^{\Lambda_n}_s)^{-1}$.
Then, for $n>m$, define
$$
f(s)= \tilde\alpha_s^{\Lambda_n}(A)- \tilde\alpha_s^{\Lambda_m}(A).
$$
and observe that
\begin{eqnarray*}
f^\prime(s)&=& i[D_{\Lambda_n}(s),  \tilde\alpha^{\Lambda_n}_s(A)] 
- i[D_{\Lambda_m}(s),  \tilde\alpha^{\Lambda_m}_s(A)]\\
&=&  i[D_{\Lambda_n}(s),  f(s)] +i[D_{\Lambda_n}(s)-D_{\Lambda_m}(s),\tilde\alpha_s^{\Lambda_m}(A)]\,.
\end{eqnarray*}
Hence,
\begin{equation}
\left\| \tilde\alpha_s^{\Lambda_n}(A)- \tilde\alpha_s^{\Lambda_m}(A) \right\| = \Vert f(s)\Vert\leq
 \int_0^s \Vert [D_{\Lambda_n}(r)-D_{\Lambda_m}(r),\tilde\alpha_r^{\Lambda_m}(A)]\|\, dr
\,.
\end{equation}
We will show that the right-hand-side goes to zero as $n,m \to \infty$.

We begin by writing the difference as
\begin{equation*}
D_{\Lambda_n}(r)-D_{\Lambda_m}(r) = \sum_{\substack{Z\subset\Lambda_n : \\ Z\cap(\Lambda_n\setminus\Lambda_m)\neq\emptyset}}\Psi_{\Lambda_n}(Z,r) + \sum_{Z\subset \Lambda_m}\left(\Psi_{\Lambda_n}(Z,r) - \Psi_{\Lambda_m}(Z,r)\right)\,.
\end{equation*}
For the first term, the Lieb-Robinson bound of Theorem~\ref{thm:lrbpe}, which clearly applies to $\tilde{\alpha}_r^{\Lambda_m}$ as well, yields 
\begin{equation*}
\|[\Psi_{\Lambda_n}(Z,r), \tilde\alpha_r^{\Lambda_m}(A)]\|\leq 2\|A\|\|\Psi_{\Lambda_n}(Z,r)\| g(r) \sum_{x\in \Lambda_0,y\in Z}F_\Psi(d(x,y)) \, .
\end{equation*}
After summing over $Z$ and integrating, we find that
\begin{multline*}
\int_0^s \sum_{\substack{Z \subset \Lambda_n: \\Z\cap(\Lambda_n\setminus\Lambda_m)\neq\emptyset}} \|[\Psi_{\Lambda_n}(Z,r), \tilde\alpha_r^{\Lambda_m}(A)]\| \, dr \, \\
\leq 2 \|A\| \int_0^s  \sum_{\substack{Z \subset \Lambda_n: \\Z\cap(\Lambda_n\setminus\Lambda_m)\neq\emptyset}} \| \Psi_{\Lambda_n}(Z,r) \| g(r) \, dr \sum_{x \in \Lambda_0, y \in Z} F_{\Psi}(d(x,y)) \\
\leq 2 \|A\| \int_0^s g(r) \, dr \, \sum_{y \in \Lambda_n , z\in\Lambda_n\setminus\Lambda_m}\sum_{\substack{Z \subset \Lambda_n: \\ z,y \in Z}} 
\sup_{0 \leq  r \leq 1} \|\Psi_{\Lambda_n}(Z,r)\| \sum_{x\in \Lambda_0}F_\Psi(d(x,y)) \\
\leq 2 \| A \| \| \Psi \| C_\Psi \int_0^s g(r) \, dr \,  \sum_{z\in\Lambda_n\setminus\Lambda_m} \sum_{x \in \Lambda_0} F_\Psi(d(x,z))
\end{multline*}
which vanishes as $m<n\to\infty$ by the uniform integrability of $F_\Psi$.

To control the second term, we arrange the set of subsets of $\Lambda_m$, which we denote by $\caP( \Lambda_m)$, as a union of
three sets: $\caP(\Lambda_m) = \caP_1 \cup \caP_2 \cup \caP_3$ where
\begin{equation}
\caP_1 = \{ Z \in \caP( \Lambda_m) : Z\subset\Lambda_{m/3}^c\},  \quad \caP_2 = \{Z \in \caP(\Lambda_m) : Z\subset\Lambda_{2m/3}\}, 
\end{equation}
and
\begin{equation}
\caP_3 = \{ Z \in \caP(\Lambda_m) : Z \cap \Lambda_{m/3} \neq \emptyset \quad \mbox{and} \quad Z \cap \Lambda_{2m/3}^c \neq \emptyset \}.
\end{equation}

We first sum over $\caP_1$. Repeating the argument we used above, in particular using the uniform Lieb-Robinson estimate for 
both $\Psi_{\Lambda_n}(Z,r)$ and $\Psi_{\Lambda_m}(Z,r)$, we find that
\begin{equation*}
\int_0^s \sum_{Z\in\caP_1}\|\left[\Psi_{\Lambda_n}(Z,r) - \Psi_{\Lambda_m}(Z,r), \tilde\alpha_r^{\Lambda_m}(A)\right]\| \, dr \,  \leq 4 \| A \| \| \Psi \| C_\Psi \int_0^s g(r) \, dr \, \sum_{z \in\Lambda_{m/3}^c } \sum_{x \in \Lambda_0} F_\Psi(d(x,z))\, ,
\end{equation*}
and this bound decays to zero as $m\to\infty$.

We next estimate the sum over $\caP_2$. We begin by trivially bounding 
\begin{eqnarray}
\left\| \left[ \sum_{Z \in \caP_2} \left( \Psi_{\Lambda_n}(Z,r)-\Psi_{\Lambda_m}(Z,r) \right), \tilde\alpha_r^{\Lambda_m}(A) \right] \right\| & \leq & 
2 \| A \| \left\| \sum_{Z \in \caP_2} \left( \Psi_{\Lambda_n}(Z,r) - \Psi_{\Gamma}(Z,r) \right) \right\|  \nonumber \\
& \mbox{ } & \quad + 2 \| A \| \left\| \sum_{Z \in \caP_2} \left( \Psi_{\Gamma}(Z,r) - \Psi_{\Lambda_m}(Z,r) \right) \right\|
\end{eqnarray}
where we are using the notation from (\ref{eq:PsiG}). Each of the terms on the right-hand-side above
will be estimated similarly. In fact, note that 
\begin{equation}
\sum_{Z \in \caP_2} \left( \Psi_{\Gamma}(Z,r) - \Psi_{\Lambda_m}(Z,r) \right) = \sum_{Z \subset \Lambda_{2m/3}} \sum_{n \geq 0} \sum_{\substack{Y \subset \Gamma: \\ Y_n=Z}}
\left( \Delta_{\Gamma}^n( \Phi_Y'(r),r) - \Delta_{\Lambda_m}(\Phi_Y'(r),r) \right)
\end{equation}
implies a bound of the form
\begin{eqnarray} \label{eq:p2bd}
\left\| \sum_{Z \in \caP_2} \left( \Psi_{\Gamma}(Z,r) - \Psi_{\Lambda_m}(Z,r) \right) \right\| & \leq & \sum_{n \geq 0} \sum_{\substack{Y \subset \Gamma: \\ Y_n \subset \Lambda_{2m/3}}}
\left\| \Delta_{\Gamma}^n( \Phi_Y'(r),r) - \Delta_{\Lambda_m}(\Phi_Y'(r),r) \right\| \nonumber \\
& \leq & 4 \sum_{n \geq 0} \sqrt{G(n-1)} \sum_{y \in \Lambda_{2m/3}} \sum_{\substack{Y \subset \Gamma: \\ y \in Y}} |Y| \sup_{0 \leq r \leq 1} \| \Phi_Y'(r) \| \sqrt{K(d( \Lambda_{2m/3}, \Lambda_m^c))} \nonumber  \\
& \leq & 4 \| \partial \Phi \|_a F_a(0) \sum_{n \geq 0} \sqrt{G(n-1)} \cdot | \Lambda_{2m/3} | \sqrt{K(b_1 m/3)} \, .
\end{eqnarray}
Since $| \Lambda_{2m/3}| \leq b_2 (2m/3)^p$, it is clear that the above goes to zero as $m \to \infty$; uniformly for $0 \leq r \leq 1$. The bound
corresponding to (\ref{eq:p2bd}) with $\Lambda_m$ replaced with $\Lambda_n$ goes to zero at least as fast.

Finally, we sum over $\caP_3$. These sets extend over a large fraction of $\Lambda_m$, and therefore, they must correspond to terms with small norms. 
Indeed, 
\begin{multline*}
\int_0^s \sum_{Z\in\caP_3} \Vert \left[\Psi_{\Lambda_n}(Z,r)-\Psi_{\Lambda_m}(Z,r), \tilde\alpha_r^{\Lambda_m}(A)\right]\Vert  dr \\
 \leq 2s\|A\| \sum_{x\in\Lambda_{m/3}} \sum_{y\in\Lambda_{2m/3}^c} \sum_{\substack{Z \subset \Gamma: \\x,y \in Z}} \left( \sup_{0\leq r \leq 1}\|\Psi_{\Lambda_n}(Z,r)\| + \sup_{0\leq r \leq 1} \|\Psi_{\Lambda_m}(Z,r) \| \right)  \\ \leq 4 s\|A\|\|\Psi\| \sum_{x \in \Lambda_{m/3}} \sum_{y\in\Lambda_{2m/3}^c} F_\Psi(d(x,y))\,.
\end{multline*}
As is proven in Theorem~\ref{thm:intdec}, the function $F_{\Psi}(r) = u_{\mu}(r) F(r)$ for some $\mu>0$ and $r$ large enough. Thus the sum 
\begin{equation}
 \sum_{x \in \Lambda_{m/3}} \sum_{y\in\Lambda_{2m/3}^c} F_\Psi(d(x,y)) \leq \| F \| | \Lambda_{m/3}| u_{\mu}(b_1 m/3) \, 
\end{equation}
which goes to zero as $m \to \infty$. 
We have shown that all terms vanish in the limit, and therefore, the sequence $( \tilde\alpha_s^{\Lambda_n}(A))_n$ is Cauchy as claimed.
\end{proof}

The above result establishes the existence of the spectral flow in the thermodynamic
limit, and we have denoted that limiting flow by $\alpha^\Gamma_s$. Arguments similar to
those used in the proof of Theorem \ref{thm:ThermoLimit} show that $\alpha^\Gamma_s$ is also
the thermodynamic limit of the flows generated by the interaction $\Psi_\Gamma(s)$,
defined in \eq{eq:PsiG}, restricted to the sequence of finite volumes $\Lambda_m$.
This is not a surprise since, as the next proposition shows, $\Psi_\Gamma(s)$ is the 
limit of $\Psi_\Lambda(s)$ as $\Lambda\to\Gamma$. In this proposition, we consider
the interactions $\Psi_\Lambda(s)$ as functions defined on the power set of $\Lambda$,
$\mathcal{P}(\Lambda)$, with values in the algebra of observables. As such, we can
consider the interactions obtained by restriction to a subset of $\mathcal{P}(\Lambda)$,
such as $\Psi_{\Lambda}(s) \vert_{\mathcal{P}(\Lambda_0)}$, for $\Lambda_0\subset\Lambda$.

\begin{prop}\label{prop:intconv}
For any finite $\Lambda \subset \Gamma$ and $Z \subset \Lambda$,
the following estimate holds
\begin{equation}\label{pointwise}
\|\Psi_{\Lambda}(Z,s)-\Psi_{\Gamma}(Z,s)\| 
\leq C \| \partial \Phi \|_a  |Z|  \sqrt{K(d(Z, \Lambda^c))} \,
\end{equation}
where
$$
C=4  F(0) \left( \sqrt{\Vert W_\gamma\Vert_1}+ \sum_{n\geq 0}\sqrt{G(n)}\right)\, .
$$
Let $\left( \Lambda_m \right)_m$ be a sequence of finite volumes  satisfying the properties 
\eq{conditionsLambda_m}. Then, for any $\beta \in(0,1)$, one has
\be\label{norm_convergence}
\lim_{m\to\infty}\Vert \Psi_{\Lambda_m}\vert_{\mathcal{P}(\Lambda_{m-m^\beta})}
-\Psi_\Gamma\vert_{\mathcal{P}(\Lambda_{m-m^\beta})}\Vert_{F_\Psi}=0
\ee
\end{prop}

\begin{proof} 
To prove the estimate \eq{pointwise} for fixed $Z$, we apply Lemma~\ref{lem:Delbd2}
with $A=\Phi_Z'(s)$ and then Lemma~\ref{lem:simple_bounds} as follows:
\begin{eqnarray}
\|\Psi_{\Lambda}(Z,s)-\Psi_{\Gamma}(Z,s)\| & \leq & \sum_{\stackrel{Y,n \geq 0:}{Y_n=Z}} \left\| \Delta_{\Lambda}^n( \Phi_Y'(s),s) - \Delta_{\Gamma}^n( \Phi_Y'(s),s) \right\| \nonumber \\
& \leq & 4  \sum_{\stackrel{Y,n \geq 0:}{Y_n=Z}} |Y| \| \Phi_Y'(s) \| \sqrt{G(n-1) K(d(Y, \Lambda^c))} \nonumber \\
& \leq & 4 \left( \sqrt{\Vert W_\gamma\Vert_1}+\sum_{n\geq 0} \sqrt{G(n)}\right)\sqrt{K(d(Z,\Lambda^c))} 
\sum_{Y\subset Z} |Y| \| \Phi_Y'(s) \|  \nonumber \\ 
 & \leq & 4 \| \partial \Phi \|_a F(0) \left( \sqrt{\Vert W_\gamma\Vert_1}+\sum_{n\geq 0} \sqrt{G(n)} \right)
 |Z| \sqrt{K(d(Z, \Lambda^c))} \, ,
\end{eqnarray}
which is the claimed result. To prove \eq{norm_convergence} is now a straightforward application
of \eq{pointwise} and the properties of the function $K$ defined in Lemma~\ref{lem:Delbd2}.
\end{proof}

\begin{prop}\label{prop:properties}
The spectral flow $\alpha^\Gamma_s$ for the infinite system has the following properties:
\begin{enumerate}
\item $\left( \alpha^\Gamma_s \right)_{s\in [0,1]}$ is a strongly continuous cocycle of automorphisms
of the $C^*$-algebra of quasi-local observables, and it is the thermodynamic limit
of the finite-volume cocycles generated by the interaction $\Psi_\Gamma(s)$.
\item $\alpha^\Gamma_s$ satisfies the Lieb-Robinson bound
\begin{equation}
\left\|  \left[ \alpha_s^{\Gamma}(A), B \right] \right\| \leq 2 \| A \| \| B \| \min \left[ 1, g(s) \sum_{x \in X, y \in Y} F_{\Psi}(d(x,y)) \right] ,
\end{equation}
for any $A \in \mathcal{A}_X$, $B \in \mathcal{A}_Y$, and $0 \leq s \leq 1$, 
with $g$ given by
\begin{equation}
C_{F_{\Psi}} \cdot  g(t) = \left\{ \begin{array}{cc} e^{2 \| \Psi \| C_{F_\Psi} |t|} - 1  & \mbox{if } d(X,Y) >0, \\
e^{2 \| \Psi \| C_{F_\Psi} |t|} & \mbox{otherwise}. \end{array} \right.
\end{equation}
and the quantities $F_\Psi$, $C_{F_\Psi}$, and $\Vert \Psi\Vert_{F_\Psi}$ as given in
Theorem~\ref{thm:lrbpe}.
\item If $\beta$ is a local symmetry of $\Phi$, \ie, an automorphism such that $\beta(\Phi(X,s))
=\Phi(X,s)$, for all $X\subset \Gamma$ and $s\in[0,1]$, then $\beta$ is also a symmetry of 
$\alpha_s^\Gamma$, \ie, $\alpha^\Gamma_s\circ\beta =\alpha^\Gamma_s$ for all $s\in [0,1]$.
\item Suppose $\Gamma$ is a lattice with a group of translations $\left( T_x \right)_x$ and $\left( \pi_{T_x} \right)_x$
is the representation of the translations as automorphisms of the quasi-local algebra
$\mathcal{A}_\Gamma$. Then, if $\Phi$ is translation invariant, \ie, $\Phi(T_x(X),s)=
\pi_{T_x}(\Phi(X,s))$, for all $X\subset\Gamma$, and $s\in [0,1]$, then 
$\alpha_s^\Gamma$ commutes with $\pi_{T_x}$, for all $x$ and $s$.
\end{enumerate}
\end{prop}

\begin{proof}
All these properties follow from the preceding results.
\end{proof}

\subsection{Automorphic equivalence of gapped ground states}
\label{sec:automorphic_equivalence}

We can now describe more precisely the problem of equivalence of quantum phases 
discussed in the introduction. Let $\mathcal{S}_\Lambda(s)$ denote the set of states of the system in volume $\Lambda$ that are mixtures of eigenstates with energy in $I(s)$ and let $\mathcal{S}(s)$ be the set of weak-$*$ limit points as $n\to\infty$ of $\mathcal{S}_{\Lambda_n}(s)$. Note that these sets are non-empty. The result of Section \ref{sec:Evolution} immediately implies
\begin{equation}
\mathcal{S}_{\Lambda_n}(s)=\mathcal{S}_{\Lambda_n}(0)\circ\alpha^{\Lambda_n}_s\,,
\label{finite_gs_relation}
\end{equation}
where $\alpha^{\Lambda_n}_s$ is the automorphism defined in (\ref{automorphism}).
In Section \ref{Sec:LRforD} we proved that $\alpha^{\Lambda_n}_s$ satisfy a 
Lieb-Robinson bound with a uniformly bounded Lieb-Robinson velocity and
decay rate outside the `light cone'. In the previous subsection we obtained the
thermodynamic limit of these automorphisms leading to the cocycle $\alpha_s^\Gamma$ which
automatically satisfies a Lieb-Robinson bound with the same estimates for the velocity and the decay. The following theorem states that~(\ref{finite_gs_relation}) carries 
over to the thermodynamic limit.

\begin{thm} \label{thm:main}
The states $\omega(s)\in\mathcal{S}(s)$ in the thermodynamic limit are automorphically equivalent to the states $\omega(0)\in\mathcal{S}(0)$ for all $s\in[0,1]$. Indeed,
\begin{equation}
\mathcal{S}(s)=\mathcal{S}(0)\circ\alpha_s^{\Gamma}
\label{infinite_gs_relation}
\end{equation}
Moreover, the connecting automorphisms $\alpha_s^{\Gamma}$ can be generated by a 
$s$-dependent quasi-local interaction $\Psi(s)$ with $\| \Psi \|_{F_{\Psi}} <\infty$, where
the norm is defined in \eq{triple_norm_psi}.
$\alpha_s^{\Gamma}$ then satisfies the same Lieb-Robinson bound as $\alpha_s^\Lambda$ in 
Theorem~\ref{thm:lrbpe}.
\end{thm}

\begin{proof}
This is a direct consequence of~(\ref{finite_gs_relation}), theorem~\ref{thm:ThermoLimit} and the lemma below.
\end{proof}

\begin{lemma}
Let $(\sigma_n)_n$ be a strongly convergent sequence of automorphisms of a $C^*$-algebra
$\mathcal{A}$, converging to $\sigma$ and let $(\omega_n)_n$ be a sequence of states
on $\mathcal{A}$. Then the following are equivalent:
\begin{enumerate}
\item $\omega_n$ converges to $\omega$ in the weak-$*$ topology;
\item $\omega_n\circ\sigma$ converges to $\omega\circ\sigma$ in the weak-$*$ topology;
\item $\omega_n\circ\sigma_n$ converges to $\omega\circ\sigma$ in the weak-$*$ topology.
\end{enumerate}
\end{lemma}
\begin{proof}
(i)$\Leftrightarrow$(ii) follows immediately from the fact that $\sigma$ and $\sigma^{-1}$
are automorphisms. Now if (ii) holds, the second term of
\begin{equation*}
\vert(\omega_n\circ\sigma_n)(A)-(\omega\circ\sigma)(A)\vert \leq \vert\omega_n(\sigma_n(A)-\sigma(A)) \vert + \vert\omega_n(\sigma(A))-\omega(\sigma(A)) \vert\,,
\end{equation*}
vanishes. So does the first one
\begin{equation*}
\vert\omega_n(\sigma_n(A)-\sigma(A)) \vert \leq \Vert\omega_n\Vert \Vert\sigma_n(A)-\sigma(A) \Vert \longrightarrow 0
\end{equation*}
since $\omega_n$ are states, and therefore (iii) holds. A similar argument yields (iii)$\Rightarrow $(ii).
\end{proof}

In the recent literature \cite{chen:2010a,chen:2010b}, a `ground state phase' has been 
defined as an equivalence class of ground states with the equivalence defined as
follows: the states $\omega_0$ and $\omega_1$ are equivalent (\ie,
belong to the same phase) if there exists a continuous family of Hamiltonians 
$H(s)$, $0\leq s\leq 1$, such that for each $s$, $H(s)$ has a gap above the ground state
and $\omega_0$ and $\omega_1$ are ground states of $H(0)$ and $H(1)$, respectively. 
As an alternative definition the authors of \cite{chen:2010a} state that $\omega_0$ 
and $\omega_1$ should be related by a `local unitary transformation'.
With Theorem \ref{thm:main} we provide precise conditions under which
the first property implies the second. At the same time we have 
clarified the role of the  thermodynamic limit left implicit in the cited works.

Based on Theorem \ref{thm:main} it seems reasonable to define the
ground states of two interactions $\Phi(0)$ and $\Phi(1)$ to be in the same
phase if there exists a differentiable interpolating family of interactions $\Phi(s)$,
$0\leq s\leq 1$, such that there exists $a>0$ for which $\Vert \Phi\Vert_a + 
\Vert\partial \Phi\Vert _a <\infty$, and if the spectral gap above the
ground states of the corresponding finite-volume Hamiltonians 
$H_{\Lambda_m}(s)$ have a uniform lower bound $\gamma>0$.
The increasing sequence of finite volumes $\Lambda_m$ should satisfy
a condition of the type \eq{conditionsLambda_m}.
One should allow for a space of nearly degenerate eigenstates of $H_{\Lambda_m}$
which, in the thermodynamic limit, converge to a set of ground states
$\mathcal{S}(s)$.
We have proved that under these conditions the sets of thermodynamic
limits of ground states are connected by a flow of automorphisms generated
by a quasi-local interaction with almost exponential decay and satisfying
a Lieb-Robinson bound. We believe that these are {\em sufficient conditions}
for belonging to the same gapped ground state phase. More work is needed
to identify {\em necessary conditions}. 

We remark that a `ground state phase' should be defined as an equivalence relation 
on simplices of states of a quantum lattice system.
This is an equivalence of sets of states rather than of models because it is possible that 
different quantum phases coexist as ground states of one model, while the same states also
appear as unique ground states of other models. Examples of this situation can 
easily be constructed using frustration free models in one dimension with finitely correlated
ground states, also known as matrix product states \cite{fannes:1992,nachtergaele:1996}.
In particular, if $\mathcal{S}(s)$ denotes the set of infinite-volume ground states of a model
with parameter $s$, the relation $\mathcal{S}(s)=\mathcal{S}(0)\circ\alpha_s$,
does not imply that the states in the sets $\mathcal{S}(s)$ are automorphically
equivalent among themselves. E.g., if for a model with a discrete symmetry we find that
symmetry broken states coexists with symmetric states, $\alpha_s$ cannot map these two 
classes into each other. In general, as emphasized in Proposition~\ref{prop:properties},
the $\alpha_s$ we constructed posses all symmetries of the Hamiltonians.

There are plenty of examples of models to which our results apply.
Clearly, the various perturbation results mentioned in the introduction provide
many interesting examples of sets of models with ground states in 
a variety of types of gapped phases. Another class of examples is provided
by the rich class of gapped quantum spin chains with matrix product ground states.
In Yarotsky's work \cite{yarotsky:2006} it is shown how perturbation theory around
a matrix product ground states can be applied  to connect these two classes of examples. 
Exactly solvable models with gapped ground states depending on a parameter,
such as the anisotropic $XY$ chain \cite{lieb1961}, is another set of examples.
More recently, stability under small perturbations of the interaction was proved
for a class of models with topologically ordered ground states \cite{bravyi:2010a};
these include e.g. Kitaev's toric code model \cite{kitaev:2003}. Our results are also applicable
to this class of models.
It seems likely that other applications will be found. As an example of an
application left to be explored, we mention that
the existence of a connecting automorphism of the type $\alpha_s$
can provide a means to distinguish true quantum phase transitions from isolated
critical (i.e., gapless) points around which it is possible to circumnavigate with
suitably chosen perturbations.

\subsection*{Acknowledgment}

This work was supported by the National Science Foundation
and the Department of Energy: 
S.B. under Grant DMS-0757581, B.N. under grant
DMS-1009502, and R.S. under Grant DMS-0757424.
S.M. received support from  NSF DMS-0757581 and 
PHY-0803371, and DOE Contract DE-AC52-06NA25396. 
BN greatfully acknowledges the kind hospitality of the Institute 
Mittag-Leffler (Djursholm, Sweden) during Fall 2010 where part of the
work reported here was carried out and of the Department of Mathematics 
at the University of Arizona where it was completed.

\end{document}